\documentclass[a4paper,UKenglish,cleveref, autoref, thm-restate]{lipics-v2021}

\graphicspath{{./figures/}}

\usepackage{mathtools}
\usepackage{thm-restate}

\usepackage{hyperref,pifont,rotating,todonotes}

\usepackage{xcolor}

\definecolor{ao(english)}{rgb}{0.0, 0.5, 0.0}

\newcommand{\rdist}{\ensuremath{\textsf{rdist}}}
\newcommand{\ring}{\ensuremath{\textsf{Ring}}}

\newcommand{\gparallel}{\ensuremath{G^\textsf{par}}}
\newcommand{\radgraph}{\ensuremath{G^\textsf{rad}}}

\newcommand{\gcontract}{G^\textsf{con}}
\newcommand{\gmod}{G^\textsf{mod}}
\newcommand{\wmod}{\mathcal W_\textsf{mod}}
\newcommand{\tw}{\textsf{tw}}

\newcommand{\cl}{\textsf{cl}}

\newcommand{\ccheck}[1]{\textcolor{black}{#1}}

\title{Linear-Time FPT Algorithm for Surface Disjoint Paths via Surface Cutting}

 \author{Kyungjin Cho}{Institute of Algorithms and Theory, Graz University of Technology, Korea}{kyungjin.cho@tugraz.at}{0000-0003-2223-4273}{Supported by the Austrian Science Fund (FWF) \url{https://doi.org/10.55776/P36280}.} 
\author{Eunjin Oh}{POSTECH, Korea}{eunjin.oh@postech.ac.kr}{https://orcid.org/0000-0003-0798-2580}{Supported by Institute of Information \& Communications Technology Planning \& Evaluation (IITP) grant funded by the Korea government (MSIT) (No.RS-2024-00440239, Sublinear Scalable Algorithms for Large-Scale Data Analysis) and the National Research Foundation of Korea (NRF) grant funded by the Korea government (MSIT) (No.RS-2024-00358505).}
\author{Sebastian Wiederrecht}{School of Computing, KAIST, Daejeon, South Korea}{wiederrecht@kaist.ac.kr}{0000-0003-0462-7815}{Supported by the Institute for Basic Science (IBS-R029-C1)}
\authorrunning{K Cho, E Oh, and S Wiederrecht}
\Copyright{Kyungjin Cho, Eunjin Oh, Sebastian Wiederrecht} 



\category{} 

\relatedversion{} 

\acknowledgements{}

\nolinenumbers 

\hideLIPIcs
\EventEditors{John Q. Open and Joan R. Access}
\EventNoEds{2}
\EventLongTitle{42nd Conference on Very Important Topics (CVIT 2016)}
\EventShortTitle{CVIT 2016}
\EventAcronym{CVIT}
\EventYear{2016}
\EventDate{December 24--27, 2016}
\EventLocation{Little Whinging, United Kingdom}
\EventLogo{}
\SeriesVolume{42}
\ArticleNo{23}

\begin{document}
\maketitle
\begin{abstract}
    We study the \textsc{$k$-Disjoint Paths} problem on a graph embedded on a surface with bounded Euler genus. Given a graph $G$ with $n$ vertices and $k$ vertex pairs embedded on a surface of Euler genus $g$, we present a $2^{O(k^2+g^2)}n$-time algorithm that computes $k$ pairwise vertex-disjoint paths connecting the given vertex pairs if such paths exist. Our approach relies on the decomposition of $G$ into $O(k+g)$ planar subgraphs while bounding the complexity of the boundaries between these subgraphs.
    This approach enables the use of techniques for compressing linkages in planar graphs. Moreover, our techniques yield two kernels of size polynomial in $k$, $g$, and the treewidth of the graph, and of size $2^{O(k+g)}$.
    These results extend recent advances on \textsc{$k$-Disjoint Paths} on planar graphs [Cho et al. SODA 2023] and  [Włodarczyk and Zehavi FOCS 2023] to surface-embedded graphs.
\end{abstract}
\newpage

\setcounter{page}{1}
\section{Introduction}
The \textsc{Disjoint Paths} problem is a fundamental routing problem defined on pairs $(G,\mathcal T)$: Here, $G$ is an undirected graph with $n$ vertices and $\mathcal T$ is a set of $k$ pairs $(s_1,t_1),\ldots, (s_k,t_k)$ of vertices of $G$.
The goal is to find $k$ pairwise vertex-disjoint paths connecting $s_i$ and $t_i$ for each $i\in\{1,\ldots, k\}$ or to conclude that such paths do not exist in $G$.
We call the vertices appearing in $\mathcal T$ the \emph{terminals}. 
Due to its numerous applications and interesting graph theoretic properties, there exists an extensive body of literature on the \textsc{Disjoint Paths} problem, and it is known to be \textsf{NP}-complete even for grid graphs~\cite{chuzhoy2018almost} when $k$ is part of the input.
To circumvent this fact, there have been two perspectives from which to approach the problem: the lens of approximation algorithms and the lens of parameterized algorithms.
For approximation algorithms, the goal is to connect as many terminal pairs as possible using pairwise vertex-disjoint paths in polynomial time.
The best-known approximation algorithm has an approximation ratio $O(\sqrt{n})$~\cite{kolliopoulos2004approximating}.
Even for planar graphs, the best-known approximation ratio is $O(n^{9/19}\log^{O(1)}n)$~\cite{chuzhoy2016improved}.
We focus on the approach through ``parameterized algorithms'', following along a route that can be traced back to the Graph Minors Series of Robertson and Seymour \cite{robertson1995graph}.

\ccheck{
When discussing the parameterized complexity of the \textsc{Disjoint Paths} problem, one usually refers to the $k$-\textsc{Disjoint Paths} problem. Here, the number of terminal pairs is fixed to be at most $k$.
Almost all known FPT algorithms for the \textsc{$k$-Disjoint Paths} problem, including recent algorithms with running times $h_1(k)\cdot n^{1+o(1)}$~\cite{korhonen2024minor} and $h_2(k)\cdot n^2$~\cite{cavallaro2026optimalboundskdisjointpaths}, largely depend on the so-called \textsl{linkage function} \cite{Robertson2009GMXXI}.
Furthermore, an almost optimal upper bound on the linkage function was given, which implies a double-exponential FPT algorithm~\cite{cavallaro2026optimalboundskdisjointpaths}.}
A well-studied parametrization for the \textsc{Disjoint Paths} problem, besides the number $k$ of terminal pairs is the \textsl{treewidth} of the input graph.
For instance, there exist a $O(\tw^3)$-approximation algorithm~\cite{ene2016routing} and a $2^{O(\tw\log \tw)}n^{O(1)}$-time algorithm~\cite{scheffler1994practical} for the \textsc{Disjoint Paths} problem, where $\tw$ denotes the treewidth of the given graph $G$.
Additionally, under the Exponential Time Hypothesis, the \textsc{Disjoint Paths} problem requires $2^{\Omega(\tw \log \tw)} n^{O(1)}$ time~\cite{lokshtanov2018slightly}, and for planar graphs, it requires $2^{\Omega(\tw)}n^{O(1)}$ time~\cite{baste2015role}.

For the case where the input graph is planar, i.\@e.\@ the \textsc{Planar $k$-Disjoint Paths} problem, significant progress has been made.
Włodarczyk and Zehavi established that the \textsc{Planar $k$-Disjoint Paths} problem admits a \emph{polynomial kernel} parameterized by the treewidth and $k$~\cite{wlodarczyk2023planar}.
A parameterized problem $\Pi$, is said to admit a \emph{kernel} parameterized by $p$ if there is a polynomial-time algorithm that, given an instance $(I,p)$ of $\Pi$, translates it into an equivalent instance $(I',p')$ of size at most $h(p)$ for some computable function $h$ depending only on $p$.
The kernel is \emph{polynomial} if $h$ is a polynomial function.
Here, equivalence means that $(I,p)$ is a \textsf{yes}-instance if and only if $(I',p')$ is a \textsf{yes}-instance.
The polynomial kernel allows us to apply the well-known XP algorithm by Schrijver~\cite{schrijver1994finding}, giving a total running time of $2^{O(k^2)}n^{O(1)}$.
Independently, Cho et al.~\cite{cho2023parameterized} achieved a total running time of $2^{O(k^2)}n$ by developing a $(k+\tw)^{O(k)}\cdot n$-time algorithm.
This is the best-known FPT algorithm for the \textsc{Planar $k$-Disjoint Paths} problem. 
These results naturally prompt the question of whether similar techniques can be applied to graphs embedded in surfaces of higher genus.

In this paper, we investigate the \textsc{$g$-Surface $k$-Disjoint Paths} problem, that is a special case of the $k$-\textsc{Disjoint Paths} problem where the graph $G$ is given alongside with an embedding on a surface $\Sigma$ with {Euler genus} $g$. 
The \emph{Euler genus} of a surface $\Sigma$ is the smallest integer $g$ such that $\Sigma$ has $g$ closed curves whose removal transforms $\Sigma$ into a polygon with $2g$ sides.
Particularly, the plane and the sphere are the surfaces with Euler genus $0$. 

Many problems that admit natural (parameterized) algorithms on planar graphs can be generalized to graphs of bounded Euler genus.
Consider for example the \textsc{Minimum Multicut} problem~\cite{erickson2011minimum,colin2017multicuts} and the \textsc{Induced Cycle} problem~\cite{kobayashi2009algorithms}.
A key result in this area is an algorithm by Mohar~\cite{mohar1999linear} that finds an embedding of a graph $G$ in a fixed surface of Euler genus $g$ in a linear FPT time if such an embedding exists.

\subparagraph*{Our contribution.}

In this paper, we present a $2^{O(k^2+g^2)}n$-time algorithm for the \textsc{$g$-Surface $k$-Disjoint Paths} problem on $(G, \mathcal{T})$ where $G$ is embedded on a surface of Euler genus $g$.
As a secondary result, we show that the problem admits a polynomial kernel parameterized by $k$, $g$, and the treewidth $\tw$.
This follows from our new ``\textsl{surface cutting technique}'' as described in \cref{sur_sec:cut_reduction}.

\begin{restatable}{theorem}{ThmMain}\label{sur_thm:main_result}
    The \textsc{$g$-Surface $k$-Disjoint Paths} problem can be solved in $2^{O(k^2+g^2)}n$ time.
    Furthermore, it admits a polynomial kernel in $(k+g+\tw)$ and a kernel of size $2^{O(k+g)}$.
\end{restatable}

Since the work of Robertson and Seymour in Graph Minors XXI~\cite{Robertson2009GMXXI}, the \textsl{linkage function} has been studied as the main tool for the disjoint paths problem, which allows the \textsl{irrelevant vertex technique}: Given
an instance of the \textsc{$k$-Disjoint Paths} problem, when the treewidth of the graph is sufficiently
large, one can delete a non-terminal vertex without affecting the outcome of the instance.
By applying the reduction, one eventually obtains an equivalent instance of bounded treewidth that is amenable to dynamic programming.
Most recently, the $2^{2^{\textsf{poly}(k)}}n^2$-time FPT algorithm was obtained by applying the reduction and dynamic programming framework~\cite{cavallaro2026optimalboundskdisjointpaths}.
Precisely, they gave an almost optimal bound on the linkage function to obtain their algorithm.
This implies that, in order to improve this running time, one would need to develop techniques that avoid explicit dynamic programming on tree decompositions.
For planar graphs, single-exponential bounds~\cite{lokshtanov2020exponential} were obtained, and later improved to a $2^{O(k^2)}n$-time algorithm~\cite{cho2023parameterized}, by avoiding dynamic programming.
Our result provides the first significant milestone in extending the approaches beyond planar graphs.

\subparagraph*{Our approach and overview.}
The main goal of our algorithm is to decompose the graph $G$ embedded on a surface $\Sigma$ of Euler genus $g$ into $O(k+g)$ planar subgraphs while preserving the embedding, see \Cref{sur_sec:cut_reduction} for details.
The following theorem summarizes our result.
For a subgraph $H$ of $G$ embedded on a subsurface $\Sigma'$ (possibly with boundary) of $\Sigma$, we call a vertex $v$ of $H$ a \emph{boundary vertex} if it lies on the boundary of $\Sigma'$, or it is adjacent to a vertex of $V(G)\setminus V(H)$ in $G$.
\begin{restatable}{theorem}{SmallNiceSubgraph}\label{sur_thm:nice_decomp_tw}
    Let $G$ be a graph embedded on a surface $\Sigma$ with Euler genus $g$ and $k$ boundary vertices.
    In $O((k+g)\tw(G)\cdot n)$ time, we can cut $\Sigma$ into $O(k+g)$ topological disks with holes such that there are at most $O(k+g)$ holes and $O((k+g)\tw(G))$ boundary vertices in total.
\end{restatable}

Note that cutting the embedded graph along several curves intersecting only a few vertices has been widely used for the plane-embedded case~\cite{cho2023parameterized,reed1995rooted} to reduce the topological complexity.
Our contribution builds upon this line of ideas and extends them to the surface-embedded setting.
\Cref{sur_thm:nice_decomp_tw} allows us to employ the full toolbox of techniques for the \textsc{Planar $k$-Disjoint Paths} problem, like weak linkage enumerating and reconstructing techniques~\cite{cho2023parameterized}.
In particular, we can extend the polynomial kernel for \textsc{Planar Disjoint Paths}~\cite{wlodarczyk2023planar}, as stated in \Cref{sur_cor:kernelization}.
Furthermore, by adapting the treewidth bound $2^{O(k+g)}$~\cite[Theorem 4.10]{cavallaro2026optimalboundskdisjointpaths}, the corollary also implies the second kernel stated in \Cref{sur_thm:main_result}, of size $2^{O(k+g)}$.

\begin{corollary}\label{sur_cor:kernelization}
    The \textsf{$g$-Surface $k$-Disjoint Paths} problem admits a kernel with $O((k+g)^{12}\tw^{24})$ vertices.
\end{corollary}
\begin{proof}
    Note that Włodarczyk and Zehavi established that in polynomial time, we can construct an {$X$-linkage-equivalent} plane-embedded graph $H'$ of size $O(|X|^{12}\tw(H)^{12})$ for the given plane-embedded graph $H$ with treewidth $\tw(H)$ and a vertex set $X\subset V(H)$~\cite[Theorem 8]{wlodarczyk2023planar}.
    Here, two graphs $H$ and $H'$ are said to be \emph{$X$-linkage-equivalent} if $X\subset V(H), V(H')$ and the \textsc{Planar $k$-Disjoint Paths} problems on $(H,\mathcal T)$ and $(H',\mathcal T)$ are equivalent for any terminal pairs $\mathcal T\subseteq X\times X$. 
    Note that the algorithm works for a planar graph $H$.
    Precisely, it starts with a fixed embedding of $H$ on a plane, and $H'$ preserves the location of vertices and the drawing of edges.
    That means, even if we start with a fixed embedding of $H$ on a plane minus open topological disks, called \emph{boundary components}, the algorithm returns a kernel $H'$ that preserves the boundary components and the fixed location of vertices.
    This property allows us to generalize the technique for the surface by decomposing the surface into several topological disks with boundary components as follows. 

    Intuitively, we can apply their algorithm to achieve $O(\tw^{24}(k+g)^{12})$ sized kernel for the \textsc{$g$-Surface $k$-Disjoint Paths} problem on $(G,\mathcal T)$ embedded on $\Sigma$.
    Particularly, we first insert the small empty boundary component into $\Sigma$ at each of the terminals in $T$, then     
    decompose $\Sigma$ into the subsurfaces $\Sigma_i$'s by \Cref{sur_thm:nice_decomp_tw}. 
    Here, $\Sigma_i$'s are topological disks with holes. 
    Furthermore, $H_i$'s are the induced subgraphs of $G$ embedded on $\Sigma_i$'s, respectively.
    Therefore, we can replace each $H_i$ by the $X_i$-linkage equivalent graph $H_i'$, where $H_i$ is the subgraph of $G$ embedded on $\Sigma_i$ and $X_i$ is the set of vertices incident to the boundary components. 
    The obtained graph $G'$ is embedded on $\Sigma$, and it has $O(\tw^{24}(k+g)^{12})$ vertices. Moreover, it is clear that $G'$ has $k$ vertex-disjoint paths connecting the pairs $\mathcal T$ if and only if $G$ has. Therefore, $(G',\mathcal T)$ is a kernel for the \textsc{$g$-Surface $k$-Disjoint Paths} on $(G,\mathcal T)$.
\end{proof}

Using \Cref{sur_thm:nice_decomp_tw}, we can extend the $(k+\tw)^{O(k)}n$-time algorithm of the \textsc{Planar $k$-Disjoint Paths} problem introduced by Cho et al.~\cite{cho2023parameterized}.
Precisely, in~\Cref{sur_sec:encoding_construction}, we present a $(k+g+\tw)^{O(k+g)}n$-time algorithm for the \textsc{$g$-Surface $k$-Disjoint Paths} problem by extending their scheme.
The treewidth can be reduced to at most $2^{O(k+g)}$~\cite{mazoit2013single,cavallaro2026optimalboundskdisjointpaths}, and we can achieve the treewidth reduction in $2^{O(k+g)}n$ time (\Cref{thm:irrelevant_surface} in \Cref{sec:irrelevant}).\footnote{Golovach et al.~\cite{golovach2025irrelevant} designed a linear-time FPT algorithm for a more general situation by allowing a large time complexity depending on the genus $g$ and $k$. We can achieve $2^{O(k+g)}n$-time by focusing on the \textsc{$g$-Surface $k$-Disjoint Paths} problem. Details are in~\Cref{sec:irrelevant}.}
Therefore, by applying the reduction, we obtain a $2^{O(k^2+g^2)}n$-time algorithm, concluding \Cref{sur_thm:main_result}.

\section{Preliminaries}\label{sec:preliminaries}

For any undefined terms and notation, we stick to~\cite{cho2023parameterized}. 
For a graph $G$, $V(G)$ denotes the vertex set of $G$, and $E(G)$ denotes its edge set of $G$. A \emph{walk} $\omega$ is a sequence of edges that join a sequence of vertices. If all vertices and edges of $\omega$ are distinct, we call $\omega$ a \emph{path}. 
Throughout this paper, we use $[n]$ to denote the set $\{1,2,\ldots, n\}$.

\subparagraph*{Surface and genus.}
A \emph{surface} $\Sigma$ is a connected compact two-dimensional manifold where each point has a neighborhood that is homeomorphic either to the plane or to the closed half-plane. Here, the \emph{boundary} of $\Sigma$ is the set of points whose neighborhoods are homeomorphic to the closed half-plane.
A surface is \emph{non-orientable} if it contains a region homeomorphic to the M\"{o}bius band, and \emph{orientable} otherwise. 
An orientable surface with \emph{orientable genus} $g$ and $b$ \emph{boundary components} is homeomorphic to a sphere where $g$ disjoint disks are removed, a handle is attached to each of the remaining $g$ circles, and then $b$ open disjoint disks are removed. 
Moreover, a non-orientable surface with \emph{non-orientable genus} $g$ and $b$ boundary components is homeomorphic to a sphere with $g$ disjoint crosscaps, instead of handles, and $b$ disjoint open disks removed. See \Cref{sur_fig:cut_operation}.
When the surface has genus zero, then it is a special case.
Especially, the surface is called a \emph{topological disk} if it is homeomorphic to a sphere with one boundary component.
For clarity, we call a surface homeomorphic to a sphere with $b$ boundary components a \emph{topological disk} with \emph{$b-1$ holes}.

We say a drawing of $G$ on $\Sigma$ is a \emph{cellular embedding} if it has no crossings on the edges and every face is homeomorphic to an open disk.
In this paper, all embeddings are assumed to be cellular embeddings.
We call the vertices in $\mathcal T$ \emph{terminals}, and usually denote by $T\subset V(G)$ the set of terminals.
Throughout the paper, we assume that $G$ has no parallel edges or loops, and a cellular embedding of $G$ on the surface $\Sigma$ is given.
Additionally, we call the vertices lying on the boundary of $\Sigma$ the \emph{boundary vertices}.

\subparagraph*{Linkages and weak linkages}For a graph $G$ embedded on a surface $\Sigma$, 
we say two walks $\omega$ and $\omega'$ are \emph{crossing} if there are four edges $e,f,e',f'$ of $G$ sharing a common endpoint such that $e, f$ are consecutive edges of $\omega$, $e',f'$ are consecutive edges of $\omega'$, and $e, e', f$, and $f'$ lie in clockwise order around their common endpoint with respect to the fixed embedding. 
For a set $\mathcal T=\{(s_1,t_1),\ldots, (s_k,t_k)\}$ of vertex pairs, a \emph{$\mathcal T$-linkage} (in a graph $G$) is a sequence $\langle P_1,\ldots, P_k\rangle $ of $k$ {vertex-disjoint} paths in $G$ such that $P_i$ connects $s_i$ and $t_i$ for {$i \in [k]$}.
A \emph{weak $\mathcal T$-linkage} $\mathcal W$ is a sequence $\langle W_1,\ldots, W_k\rangle $ of $k$ non-crossing walks in $G$ such that $W_i$ connects $s_i$ and $t_i$ for $i\in[k]$. That means distinct walks $W_i$ and $W_j$ may share a vertex (including terminals) and edges.
We sometimes call a $\mathcal T$-linkage and a weak $\mathcal T$-linkage for a set $\mathcal T$
simply a \emph{linkage} and a \emph{weak linkage}, respectively, when we do not need to specify $\mathcal T$. 




\subparagraph*{Treewidth.}
A {\emph{tree decomposition}} of a graph $G$ is defined as a pair $(\tau,\beta)$, where $\tau$ is a tree and $\beta$ is a mapping from the nodes of $\tau$ to subsets of $V(G)$ (called the \emph{bags}) satisfying the following properties.
Let $\mathcal {B}\coloneqq \{\beta(t) \colon t\in V(\tau)\}$ be the set of bags of $\tau$.
	\begin{itemize}\setlength\itemsep{0.1em}
		\item  For any vertex $u\in V(G)$, there is at least one bag in $\mathcal {B}$ which contains $u$.
		\item For any edge $uv\in E(G)$, there is at least one bag in $\mathcal {B}$ which contains both $u$ and $v$.
		\item For any vertex $u\in V(G)$, the nodes of $\tau$ containing $u$ in their bags are connected in $\tau$.
	\end{itemize}
	The \emph{width} of a tree decomposition is defined as the size of its largest bag minus one, and the \emph{treewidth} of $G$, denoted by $\tw(G)$, is the minimum width over all tree decompositions of $G$.
    For clarity, we simply denote it by $\tw$ if it is clear from the context.
\medskip

\subsection{Radial Distance and Radial Curves}\label{sur_sec:radial_distance_curves}
Recall that the embedding of $G$ on $\Sigma$ is fixed.
A curve on $\Sigma$ is said to be a \emph{radial curve} if it intersects $G$ in vertices only. 
The {\emph{complexity}} of a radial curve is defined as the number of vertices of $G$ it intersects.
We call a simple closed radial curve a \emph{noose}.

The \emph{radial distance} between two faces $F_1$ and $F_2$ of $G$ is defined as the minimum length of a sequence of faces starting from $F_1$ to $F_2$ such that every two consecutive faces of this sequence share a common vertex. 
Moreover, the radial distance between two vertices $u$ and $v$ is defined as the minimum radial distance between two faces $F_u$ and $F_v$ incident to $u$ and $v$, respectively. 
We denote the radial distance between $u$ and $v$ as $\textsf{rdist}(u,v)$.
Note that there is a radial curve connecting $u$ and $v$ of complexity $\textsf{rdist}(u,v)+1$ but no radial curve of complexity $\textsf{rdist}(u,v)$.
We call such a curve a \emph{shortest radial curve} connecting $u$ and $v$.
For two subsets $X$ and $Y$ of $V$, we define their radial distance
as the Hausdorff distance: $\rdist(X,Y)=\max\{ \max_{x\in X}\min_{y\in Y}\rdist(x,y),  \max_{y\in Y}\min_{x\in X}\rdist(x,y)\}$.

\begin{figure}
    \centering
    \includegraphics[width=0.8\textwidth]{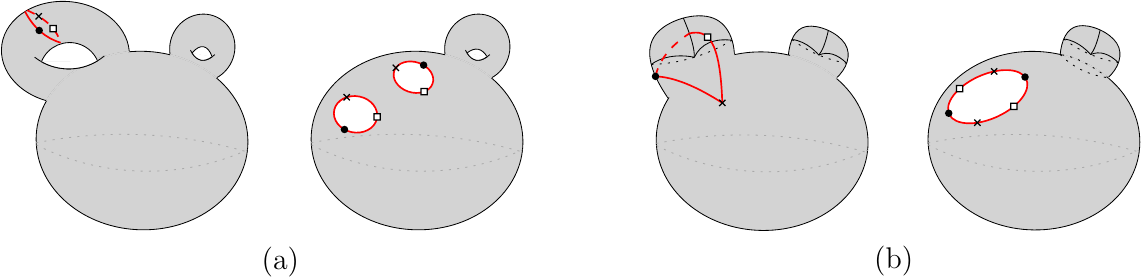}
    \caption{\small 
    Three different marks represent three vertices. (a) A cut along a non-contractible and non-separating noose on an orientable surface with two handles.
    (b) A cut along a noose on a non-orientable surface with two crosscaps whose neighborhood is homeomorphic to a M\"{o}bius band.
    \vspace{-1em}
    }
    \label{sur_fig:cut_operation}
\end{figure}

\subparagraph*{Contractible and separating curves.}
  A noose or a cycle $J$ of $G$ on $\Sigma$ is said to be {\emph{separating}} if its removal from $\Sigma$ provides at least two connected components, otherwise, it is said to be {\emph{non-separating}}.
  We say a region of $\Sigma$ is \emph{enclosed by $J$} if it is the closure of a connected component of $\Sigma-J$.
  Furthermore, we say $J$ is \emph{contractible} if it bounds a topological disk (or a topological disk minus several open disks) on the surface. Otherwise, it is \emph{non-contractible}. 
  We use $\textsf{cl}(J)$ to denote the closed region enclosed by $J$ homeomorphic to a topological disk. 
  If the surface is a sphere, and $J$ obtains two topological disks, then we choose an arbitrary one for $\textsf{cl}(J)$.
  Moreover, we call the subgraph of $G$ embedded on $\textsf{cl}(J)\setminus J$ {the} \emph{inside of} $J$.

It is a well-known property that by cutting along a non-separating noose, we can reduce the {Euler genus} of the surface~\cite[Section~6.3]{erickson2002optimally}.
Moreover, for a separating curve $J$, the sum of the Euler genus of the connected components enclosed by $J$ is equal to the Euler genus of the original surface as stated in \Cref{sur_lem:separating_genus_reducing}.
\begin{lemma}[{\cite[Lemma 2.2]{fuladi2023short}}]\label{sur_lem:separating_genus_reducing}
    Let $J$ be a separating (closed) curve on a surface $\Sigma$ and let $\Sigma_1,\dots,\Sigma_{\ell}$ be the regions enclosed by $J$.
    We have $\textsf{eg}(\Sigma)$ is equal to the sum of all $\textsf{eg}(\Sigma_i)$'s.
\end{lemma}

When a graph $G$ embedded on $\Sigma$ and a radial curve $J$ are given, we define the \emph{cut operation} of $G$ along $J$ intuitively. Refer to \Cref{sur_fig:cut_operation}.
For each vertex $v$ on $J$, we split $v$ into two vertices $v_1$ and $v_2$, then we cut all edges incident to $v$, and we reconnect them to $v_1$ and $v_2$ so that every edge is incident to the same vertex that is embedded on the same side with respect to $J$ near $v$.
Note that if $J$ is not a separating noose, it is possible that the resulting graph $G$ is not separated by cutting along $J$.


\subparagraph*{Discrete homotopy.}
The \emph{discrete homotopy} relation is a variant of the standard homotopy on embedded graphs. 
When we deal with discrete homotopy, we always work with the {so-called} {``radial completion''} of $G$. 
The \emph{radial completion} of $G$, denoted by $G^\textsf{rad}$, is a supergraph of $G$ constructed as follows.
We add one vertex for each face $F$ of $G$ and join this vertex and all the vertices incident to $F$.
Note that every face of $G^\textsf{rad}$ is a triangle.
We can consider a radial curve of complexity $N$ as a path in $G^\textsf{rad}$ of length $\Theta(N)$.
Throughout this paper, we use the bound $\tw(\radgraph) = O(\tw(G))$ rather than referring to $\tw(\radgraph)$ explicitly for clarity.

Two weak linkages $\mathcal W$ and $\mathcal W'$ are \emph{discretely homotopic} to each other if one can {be} obtained from the other by a sequence of {so-called ``face operations'' as defined below}.
There are three types of face operations: \textsf{Face Move}, \textsf{Face Pull}, and \textsf{Face Push}.
Each operation is applied to a face $F$ of $\radgraph$ and a walk $W$ of the weak linkage $\mathcal W$. Let $\partial F$ be the cycle consisting of the edges incident to $F$. 
When $\partial F\cap W$ is a subpath of $W$ while $\partial F\setminus W$ is nonempty and does not appear in any walks of $\mathcal W$, then the \textsf{Face Move} operation replaces the subwalk $\partial F\cap W$ {by} $\partial F\setminus W$.
When $\partial F$ is a subwalk of $W$, then the \textsf{Face Pull} operation removes the subwalk $\partial F$ from $W$.
When no edge of $\partial F$ appears in any walks of $\mathcal W$ while there {appear} two consecutive edges $e$ and $e'$ in $W$ {sharing} a common endpoint $v$ incident to $F$, the \textsf{Face Push} operation inserts $\partial F$ to $W$ starting and terminating at $v$ between $e$ and $e'$.
Note that the collection of walks obtained through these operations is still a weak linkage. 

\subsection{The Planar Disjoint Paths Algorithm}\label{sur_sec:preliminary_pdp}
Cho et al.~\cite{cho2023parameterized} developed a $2^{O(k^2)}n${-time} algorithm for the \textsc{Planar $k$-Disjoint Paths} problem on $(G,\mathcal T)$.
Their main approach is enumerating $(k+\tw(G))^{O(|T|)}$ weak linkages $\mathcal{L}$ connecting a terminal set $T$ so that if $G$ has a $\mathcal T$-linkage for terminal pairs $\mathcal T\subseteq T\times T$ with $|\mathcal T|\leq k$, then $\mathcal{L}$ contains a weak linkage homotopic to some $\mathcal T$-linkage.
Their approach is naturally extended to the plane with $b$ boundary components.
Briefly, by identifying each boundary component into a single terminal with multiplicity, 
the problem is reduced to have $b$ terminals on a plane with no holes, i.e., no boundary component. Refer to~\Cref{sur_fig:framed_regions}(a).
Precisely, for an identified terminal $v$, it is enough to allow a multiplicity twice the number of boundary vertices identified into $v$.
Therefore, by applying the algorithm of Cho et al.~\cite{cho2023parameterized}, we can enumerate $M^{O(b)}$ weak linkages connecting the boundaries,
where $M$ is the number of boundary vertices plus $\tw(G)$.
In this section, we provide a slightly more in-depth summary of this scheme.
More details are in~\Cref{ap:enumerating}.

\begin{figure}
    \centering
    \includegraphics[width=0.85
    \textwidth]{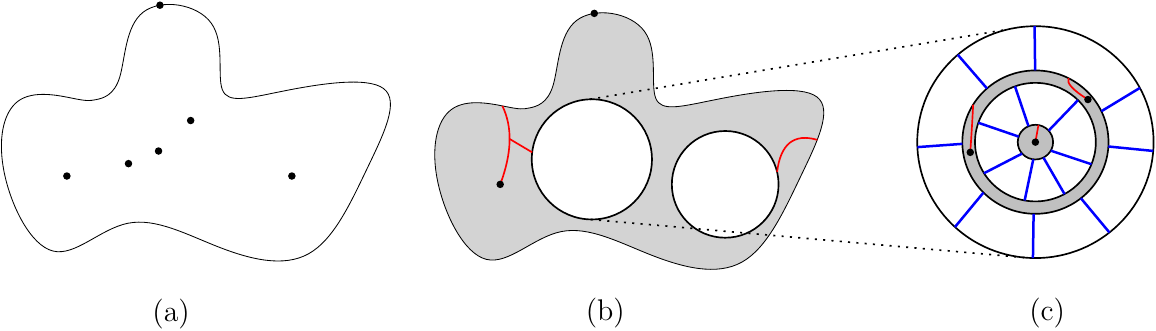}
    \caption{\small 
    (a) An embedding of a graph $\widetilde G$ in the plane. The points represent the boundary vertices.
    (b-c) 
    The white and gray areas represent boundary-free annuli and boundary-containing regions, respectively.
    The red and blue lines illustrate the skeleton forests and reference paths, respectively.
    \vspace{-1em}
    }
    \label{sur_fig:framed_regions}
\end{figure}

\subparagraph*{Enumerating weak linkages.}
Let $G$ be a planar graph embedded on the plane with $b$ boundary components.
We consider a modified graph $\widetilde G$ (and $\widetilde{\radgraph}$) obtained from $G$ (and $\radgraph$) by identifying each boundary component into a single boundary vertex. 
Note that $\widetilde G$ has $b$ boundary components, but it is embedded on a plane without a hole. Refer to \Cref{sur_fig:framed_regions}(a).
Thus, we can apply the approach of Cho et al.~\cite{cho2023parameterized} with respect to $\widetilde G$.

In the following, we use the bound $\tw(\widetilde G ) \leq \tw(G)$ rather than referring to $\tw(\widetilde G )$ explicitly for clarity.
By applying the algorithm in~\cite{cho2023parameterized}, we can construct the following structures: \emph{frames, reference paths, and skeleton forests} with respect to $\widetilde {\radgraph}$ as follows.
See \Cref{sur_fig:framed_regions} for an illustration.
The \emph{frames} are $O(b)$ nooses each with complexity $O(\tw (G))$ on $\Sigma$ which are pairwise vertex-disjoint and do not intersect any boundary vertex. 
Frames decompose the plane $\Sigma$ into $O(b)$ regions, each of which satisfies one of the following:
\begin{itemize}
\item The region contains no boundary vertex and is homeomorphic to an annulus bounded by two frames where the radial distance between the frames is at least $\Omega(\tw(G))$ within the annulus, or
\item The boundary vertices and frames incident to the region are pairwise separated by at most $O(\tw(G))$ in radial distance.
\end{itemize}
We call these two kinds of regions \emph{boundary-free annuli} and \emph{boundary-containing regions}, respectively.
Note that even if we uncontract the boundary vertices of $\widetilde G$ (and $\widetilde{\radgraph}$), a boundary-containing region fully includes or excludes any boundary component of $G$ since frames are disjoint from boundary vertices.

For a boundary-containing region, there is a collection of shortest radial curves so that cutting the region along the curves transforms it into a topological disk where all boundary vertices and frames appear on the single boundary.
We call the union of those radial curves the \emph{skeleton forest} in the region. The skeleton forest for all boundary-containing regions consists of $O(b)$ radial curves, each of complexity $O(\tw(G))$.
For each boundary-free annulus, Cho et al.~\cite{cho2023parameterized} constructed $O(\tw(G))$ pairwise vertex-disjoint paths connecting two bounding frames, called \emph{reference paths}, satisfying the following lemma.
Here, $\widetilde{\radgraph}$ is the graph obtained from $\radgraph$ by identifying each boundary component into a single vertex. 
Note that it is a supergraph of the radial completion of $\widetilde G$.
Furthermore, the frames and a skeleton forest are cycles and a forest, respectively, in $\widetilde{\radgraph}$.

\begin{restatable}{lemma}{LemSummaryPDP}\label{sur_lem:summary_pdp}
    When a planar graph $G$ is embedded on a plane with $b$ boundary components,
    we let  $M$ be the number of boundary vertices of $G$ plus $\tw(G)$.
    Then we can construct frames, skeleton forests, and reference paths, and enumerate $M^{O(b)}$ weak linkages $\mathcal W$ connecting boundary vertices on $\widetilde{\radgraph}$ in $M^{O(b)}\cdot O(n)$ time such that:
    \begin{itemize}
        \item $\mathcal W$ is a set of pairwise vertex-disjoint walks in $G$ except on frames or skeleton forests, 
        \item For a boundary-free annulus $\circledcirc$, $\mathcal W$ traverses in $\circledcirc$ along the reference paths, and
        \item If $\mathcal W$ traverses an edge in $\widetilde{G^{\textsf{rad}}}$ but not in $G$, then it is an edge on a frame or a skeleton forest.
        Additionally, such an edge is traversed at most $O(\tw(G))$ times by $\mathcal W$.
    \end{itemize}
    
    Furthermore, for any collection of pairs $\mathcal T$ of boundary vertices in $G$, some $\mathcal T$-linkage of $G$, if exists, is discretely homotopic to a weak linkage $\mathcal W$ from the enumerated weak linkages.
\end{restatable}

The above lemma summarizes the result obtained by applying the arguments from Sections~5--7 of~\cite{cho2023parameterized} to the graph $\widetilde G$. For completeness, we provide the details in~\Cref{ap:enumerating}. Briefly notice that, the original statements established by Cho et al.~\cite{cho2023parameterized} use $2^{O(|\mathcal T|)}$ instead of the $\tw(G)$.
This is because the authors achieved a bound on $\tw(G)$ using the so-called \emph{irrelevant vertex technique}~{\cite[Theorem 4.1]{cho2023parameterized}}. 
However, we can naturally derive the above lemma from their algorithms and proofs.

\begin{figure}
    \centering
    \includegraphics[width=0.75
    \textwidth]{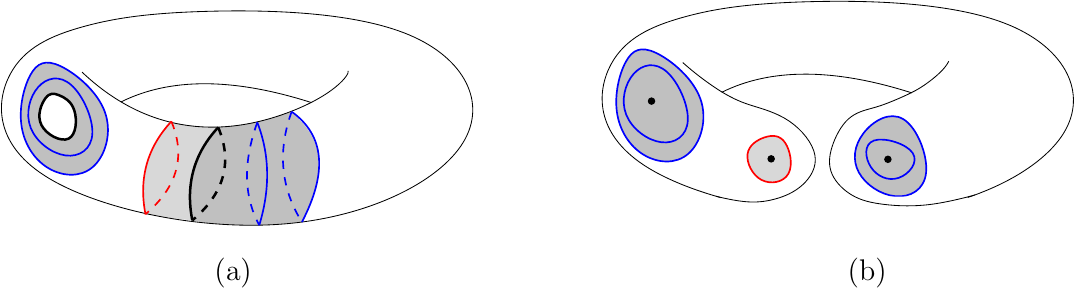}
    \caption{\small 
   The homeomorphic noose (in (a)) to a noose or boundary component in $\Sigma$ corresponds to the contractible noose (in (b)) in the modified surface isolating the contracted vertex.
   \vspace{-1em}
    }
    \label{sur_fig:noose_vertex}
\end{figure}

\section{Surface Cutting}\label{sur_sec:cut_reduction}
In this section, we enforce that our surface $\Sigma$ has $O(k)$ boundary components by adding a small empty hole for each of the terminals in $T$. The boundary vertices consist of the terminals in $T$.
We aim to cut the graph $G$ and the surface $\Sigma$ along $O(k+g)$ curves so that each resulting subsurface has Euler genus zero while bounding the complexity of the boundary components.
Particularly, we prove the following theorem, which can be seen as our main contribution and whose algorithm is a crucial step that allows for the application of the full power of techniques originally developed for plane-embedded graphs.

\SmallNiceSubgraph*

To prove the theorem, we require an algorithm computing the \emph{farthest} homotopic noose to an arbitrary noose $J$ with respect to the radial distance on $\Sigma$. 
We design such an algorithm by reducing the problem to finding a longest sequence of pairwise disjoint contractible cycles ``isolating'' a single vertex.
Specifically, we identify the vertices on $J$ into a single vertex $v$ and cut the surface at $v$.
Then we apply our isolating cycles algorithm to $v$. \Cref{sur_fig:noose_vertex} illustrates this reduction.
\Cref{sur_sec:isolating} explains how we find such a longest sequence in $O(|E(G)|)$ time. 
We then prove \Cref{sur_thm:nice_decomp_tw} in \Cref{sur_sec:prove_cutreduction} by applying this algorithm.

\subsection{Finding a Longest Sequence Isolating a Vertex}\label{sur_sec:isolating}
We say a contractible cycle (or noose) $C$ \emph{isolates} a vertex $v$ if its enclosing region contains $v$ but excludes all boundary components.
Moreover, we say a vertex $v$ is $\ell$-isolated if there is a sequence $\langle C_1,\ldots, C_{\ell}\rangle$ of cycles isolating $v$. Here, $C_1,\ldots, C_\ell$ are said to be \emph{concentric} if they are pairwise vertex-disjoint contractible cycles with $\cl(C_{i})\subsetneq \cl(C_{i+1})$ for $i\in[\ell-1]$.
In this section, we describe how to find a longest sequence of concentric cycles {isolating} $v$.

We say a sequence $\langle C_1,\ldots, C_{\ell}\rangle$ of concentric cycles is \emph{tight} {if} there is no other contractible noose $C$ lying in $\textsf{cl}(C_{i+1})\setminus \textsf{cl}(C_i)$ so that $\textsf{cl}(C_i)\subsetneq\textsf{cl}(C)\subsetneq \textsf{cl}(C_{i+1})$ for any $i\in[\ell-1]$.  
Note that a vertex $v$ is $\ell$-isolated if and only if there is a tight sequence of $\ell$ concentric cycles isolating $v$.
The remainder of this section is devoted to proving the following lemma.

\begin{restatable}{lemma}{LemIsolatedLinear}\label{sur_lem:isolated_linear}
    For a vertex $v$, we can compute a longest tight sequence $\langle C_1,\ldots, C_{\ell}\rangle$ of concentric cycles in $G$ isolating $v$ in {$O(|E(G)|)$} time.    
\end{restatable}

We describe how to compute $C_i$ if it exists by assuming that $C_{i-1}$ is given for $i\in [\ell]$.
Initially, we set $C_0$ as the single vertex $v$.
To compute $C_i$, we first contract $\textsf{cl}(C_{i-1})$ {into a single} vertex, and {refer to the resulting vertex as $v$}.
Recall that $C_{i-1}$ is a contractible cycle.
Let $F$ be the subsurface of $\Sigma$ which is the union of the faces of $G$ incident to $v$. 
Note that the boundary components of $F$ are cycles in $G$.
If a contractible cycle $C$ isolates $v$ among the boundary components of $F$, we set {$C_i\coloneqq C$}, otherwise we terminate.
\Cref{sur_claim:inductively_isolated} guarantees the tightness, and this implies the correctness of our algorithm.

\begin{figure}
    \centering
    \includegraphics[width=0.7
    \textwidth]{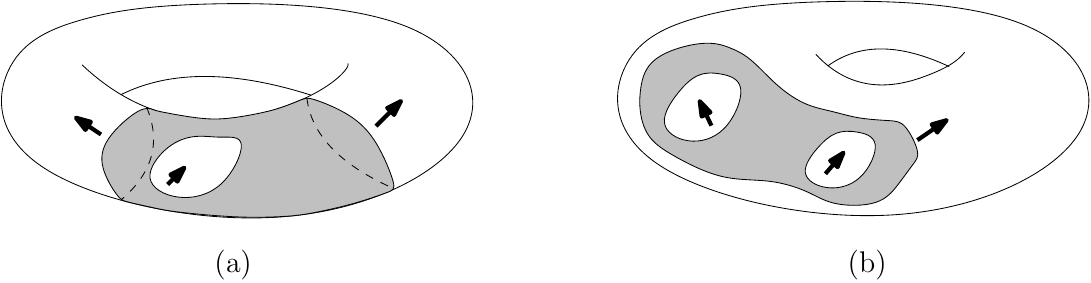}
    \caption{\small 
    The gray region represents $F$.
    (a) $F$ has three boundary components, but two of them are non-separating, and the searches starting from the components meet at some point, thus, the algorithm halts.
    (b) $F$ has three boundary components {all of which} are contractible, furthermore, one contains $F$ inside.
    Thus, our algorithm returns the boundary component as $C_i$.
    \vspace{-1em}
    }
    \label{sur_fig:multiple_bfs}
\end{figure}

\begin{restatable}{claim}{ClaimInductiveIsolated}\label{sur_claim:inductively_isolated}
    Let $C$ be a contractible cycle that is disjoint from $v$.
    Then $\textsf{cl}(C)$ contains $v$ if and only if $F$ is contained inside $\textsf{cl}(C)$, where $F$ is the union of the faces of $G$ incident to $v$.
\end{restatable}
\begin{proof}
The `if' direction is clear since $v$ refers to a single vertex that is incident to every face of $F$.
We prove the `only if' direction.
For this, we fix a contractible cycle $C$ of $G$ {whose inside contains $v$}, and we show that {any face $F'\subseteq F$ of $G$ must be contained in $\cl(C)$.}
Let $u$ be the facial vertex in the radial completion $\radgraph$ which corresponds to $F'$.
Observe that $u$ and the vertex $v$ are adjacent in $\radgraph$ by the definition of $F$.

Since $C$ is a separating cycle on $\Sigma$, no edge of {the} radial completion $\radgraph$ connects two vertices each in $\textsf{cl}(C)\setminus C$ and $\Sigma\setminus \textsf{cl}(C)$. 
By the construction, $v$ is in $\cl(C)$ and $u$ is adjacent to $v$ in $\radgraph$.
Therefore, $\cl(C)$ contains $u$ and the face $F'$, which completes the proof.
\end{proof}

When a boundary component
of $F$ is contractible and isolates $v$, then all boundaries of $F$ are contractible by \Cref{sur_claim:inductively_isolated}.
If no such contractible boundary component exists, then the algorithm halts. 
In the following, we describe how to check if all boundary components of $F$ are contractible, and then show how to choose the one isolating $v$ if it exists.

Here, we contract all boundary components of $F$ into loops for convenience. We conduct $x$ simultaneous depth-first searches on $G$ starting from each boundary component of $F$ towards the opposite of $F$, where $x$ is the number of boundary components of $F$.
Each search starts from the vertex on each boundary component of $F$, see \Cref{sur_fig:multiple_bfs}.

\subparagraph*{Halting conditions and Euler characteristic.}
Let $x$ be the number of boundary components of $F$.
Note that all boundary components of $F$ are contractible if and only if they decompose the surface $\Sigma$ into $(x+1)$ subsurfaces (including $F$) and they all have Euler genus zero except one.
Therefore, we can terminate the algorithm immediately if two searches started from distinct components visit the same vertex. 
Now, assuming this does not happen, we let $H_1,\ldots, H_{x}$ be the subgraphs of $G$ traversed by the searches.
In particular, all boundary components are separating, and $H_1,\ldots, H_x$ are embedded on subsurfaces $\Sigma_1,\ldots, \Sigma_x$ of $\Sigma$ such that $F$ and $\Sigma_1,\ldots, \Sigma_x$ are pairwise interior-disjoint and their union is $\Sigma$.

During the traversal, we compute the genus $\textsf{eg}(\Sigma_1),\ldots ,\textsf{eg}(\Sigma_x)$ by computing the \emph{Euler characteristic} $\chi(\Sigma_j)=|V(H_j)|-|E(H_j)|+|F(H_j)|+|B(H_j)|$ of $H_j$ for $j\in[x]$, where $V(H_j), E(H_j), F(H_j)$, and $B(H_j)$ are the number of vertices, edges, faces, and boundary components of $H_j$ on $\Sigma_j$, respectively. 
The Euler genus of the subsurface $\Sigma_j$ is defined as $\textsf{eg}(\Sigma_j)=2-\chi(\Sigma_j)$. 
If some $\Sigma_j$ has Euler genus $g$, then its corresponding boundary component $C$ of $F$ is contractible and includes $v$ inside. Refer to \Cref{sur_fig:multiple_bfs}(b). 
Thus, if such a subsurface $\Sigma_j$ contains all boundary components of $\Sigma$, then $C$ isolates $v$.
In this case, we set $C_i\coloneqq C$, and we move to the next iteration.
In the other case, there is a non-contractible boundary component of $F$, or no boundary component of $F$ isolates $v$.
Therefore, there is no contractible cycle isolating $v$ by \Cref{sur_claim:inductively_isolated}, and our algorithm halts.

Note that each iteration described above takes $O(|E(G)|)$ time to compute the respective $C_i$. Thus, the algorithm takes $O(|E(G)|\ell)$ time to find $\ell$ tight concentric cycles.
Since $\ell$ is in $O(n)$, the algorithm runs in quadratic time in the worst case.
{However, we can achieve linear time as shown below.}
Briefly, we instead abort when $(x-1)$ depth-first searches are halted among $x$ simultaneous depth-first searches, and we compute the genus of the subsurface $F$ of $\Sigma$.
Then, we compute the genus for all subsurfaces $\Sigma_1,\ldots, \Sigma_x$ according to \Cref{sur_lem:separating_genus_reducing}. 

\subparagraph*{Linear-time.}
Originally, our algorithm halts after all $x$ simultaneous depth-first searches are halted.
To improve efficiency, we instead abort when $(x-1)$ depth-first searches are halted, and we compute the genus of the subsurface $F$ of $\Sigma$.
Then we can compute the genus of $(x-1)$ regions among $\Sigma_1,\ldots, \Sigma_x$, and the genus of $F$.
However, since the sum of $\textsf{eg}(\Sigma_1), \ldots, \textsf{eg}(\Sigma_x),$ and $ \textsf{eg}(F)$ is $g$ by \Cref{sur_lem:separating_genus_reducing}, the algorithm provides all of the genuses.
Thus, we can check whether all boundary components of $F$ are contractible or not, and return the same solution as the original.
In the following, we analyze the running time.

We let $C_1,\ldots, C_\ell$ be the returned concentric cycles isolating $v$, and $G_i$ be the induced subgraph embedded on $\textsf{cl}(C_{i})\setminus \textsf{cl}(C_{i-1})$ for $i\in[\ell]$.
We show that to compute $C_i$ for $i\in[\ell-1]$, the modified algorithm takes $O(|E(G_i)|)$ time. 
Although we perform $x$ simultaneous depth-first searches, the algorithm aborts the last one when the other searches terminate.
If the aborted search traverses the outside of $\textsf{cl}(C_i)$, then the algorithm naturally takes $O(|E(G_i)|)$ time.
In the other case, the complexity of the subgraph $G_i$ is at least the complexity outside of $\textsf{cl}(C_{i})$, then the {total number of operations our algorithm takes is} at most $O(|E(G_i)|)$.
{Since $\textsf{cl}(C_{i})\setminus \textsf{cl}(C_{i-1})$ includes $F$, computing $\textsf{eg}(F)$ also runs in this time bound. }
Therefore, the total running time is $O(|E(G)|)$ time, which completes the proof of \Cref{sur_lem:isolated_linear}.

   \begin{figure}
    \centering
    \includegraphics[width=0.7
    \textwidth]{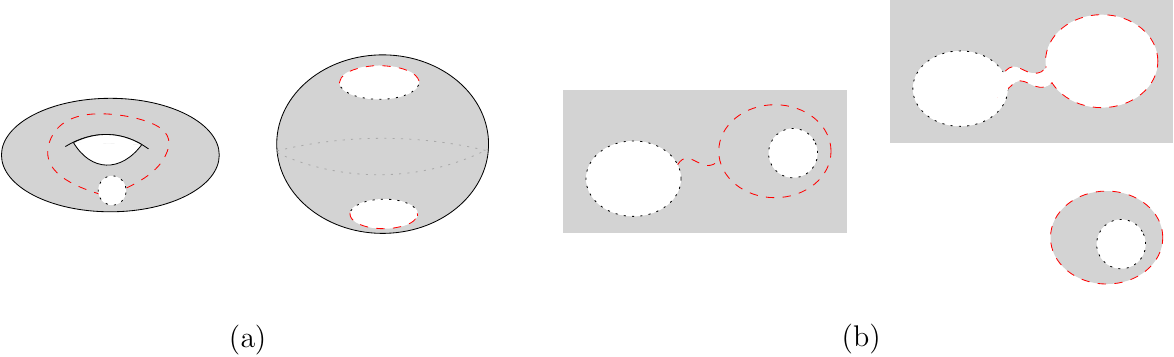}
    \caption{\small 
    (a) By cutting a non-separating and non-contractible noose, we obtain a surface with a reduced Euler genus.
    (b) A contractible noose decomposes the surface into two components so that one is embedded in the plane, and the other has fewer boundary components.  }
    \label{sur_fig:cut_reduction}
\end{figure}

\subsection{Cut Reduction}\label{sur_sec:prove_cutreduction}
For a subgraph $H$ of $G$, we define a vertex of $H$ to be a \emph{boundary vertex} if it is on the boundary of $\Sigma$ or adjacent to a vertex of $V(G)\setminus V(H)$ in $G$. 
Note that $H$ is embedded on a subsurface $\Sigma'$ of $\Sigma$ such that all boundary vertices of $H$ lie on the boundary of $\Sigma'$.
We choose $\Sigma'$ as the inclusion-wise minimal among such subsurfaces, then we let \emph{boundary component} of $H$ be the set of boundary vertices of $H$ lying on the same boundary component of $\Sigma'$.

In this section, we prove \Cref{sur_thm:nice_decomp_tw}.
Precisely, we describe how to find radial curves of complexity $\tw(G)$ which decompose $G$ into $O(k+g)$ subgraphs that are plane-embedded.
Initially, we are given $G$ embedded on $\Sigma$ with Euler genus $g$ and $k$ boundary components.
In each iteration, 
we call \Cref{sur_lem:cut_reduction} and cut the surface and graph to reduce the Euler genus or the number of boundary components (see \Cref{sur_fig:cut_reduction}).
Here, $\Sigma'$ be the inclusion-wise minimal subsurface of $\Sigma$ such that $H$ is embedded on $\Sigma'$ and its boundary vertices are on the boundary of $\Sigma'$.

\begin{restatable}{lemma}{LemCutRedcution}\label{sur_lem:cut_reduction}
    Let $H$ be the subgraph of $G$ embedded on a subsurface $\Sigma'$ of $\Sigma$ with genus $g'\leq g$ and $b'$ boundary components. If $g',b'\geq 1$, we can compute a radial curve $J$ on $\Sigma'$ of complexity $O(\tw(G))$ and in time $O(\tw(G)\cdot n)$ such that $J$ encloses at most three regions of $\Sigma'$, and one of the following holds:
    \begin{enumerate}
        \item Every region of $\Sigma'$ enclosed by $J$ has Euler genus at most $g'-1$ or
        \item A region of $\Sigma'$ enclosed by $J$ either \textsf{(i)} has fewer than $b'$ boundary components and Euler genus $g'$, or \textsf{(ii)} has at least two boundary components and Euler genus zero.
    \end{enumerate}
\end{restatable}

\begin{proof}
    We show how to find the desired radial curve $J$ on $\Sigma'$.
    For this, we first fix an arbitrary boundary component $B$ of $\Sigma'$ and contract the boundary vertices on $B$ into a single vertex $o$. We denote the resulting graph by $H_o$.
    We then compute the longest tight sequence $\langle C_1,\ldots, C_\ell\rangle$ of concentric cycles isolating $o$ from the boundary components of $\Sigma'$ except $B$ using \Cref{sur_lem:isolated_linear}.
    In the following, we further contract the vertices of $H_o$ on $\cl(C_\ell)$.
   
    Let the region $F$ in $\Sigma'$ be the union of the faces of $H_o$ incident to $o$.
    Note that there is no contractible cycle whose inside contains $o$ but whose boundary components are in $\Sigma'$ (except $B$).
    Therefore, there are two possible cases: $F$ has a non-contractible noose, or all boundary components of $F$ are contractible but one contains another boundary component $B'\neq B$ of $\Sigma'$ inside.
    We consider each case one by one.

    \subparagraph*{$F$ has a non-contractible noose.}
    We first claim that there exists a non-contractible noose of complexity at most three and intersecting $o$ if there is a non-contractible $\gamma$ in $F$, refer to \Cref{sur_fig:lateral}(a-b). 
    Note that we can compute the shortest non-contractible noose $J$ on $F$ intersecting $o$ in $O(|E(H^{\textsf{rad}})|)$-time~\cite[Lemma 5.2]{erickson2002optimally}.
    Recall that cutting along a non-contractible noose reduces the Euler genus by at least one~\cite[Secion~6.3]{erickson2002optimally}.
    Therefore, the obtained non-contractible $J$ satisfies the \textsf{Condition~1} of \Cref{sur_lem:cut_reduction}.
    According to \Cref{sur_claim:uncontracting}, we obtain a radial curve $J'$ from $J$ so that $J'$ intersects $O(\tw(G))$ vertices in $H$.

    
    \begin{restatable}{claim}{ClaimUncontracting}\label{sur_claim:uncontracting}
        We can translate the non-contractible noose $J$ into a radial curve $J'$ intersecting $O(\tw(G))$ vertices in $H$ whose removal reduces the genus of $\Sigma'$ in time $O(\tw(G)\cdot n)$.
    \end{restatable}
    \begin{proof}
    By uncontracting $o$ as the vertices in $\cl(C_\ell)$, we can intuitively translate the noose $J$ as a radial curve $J'$ connecting two boundary vertices on $B$ such that $J'$ intersects at most $2\ell+3$ vertices on $H$ and the cutting along $J'$ reduces the genus of $\Sigma$. See \Cref{sur_fig:cut_reduction}(a).
        If $\ell\in O(tw(G))$, then the claim holds. In the following, we suppose that $\ell$ is at least $\omega(\tw(G))$.
        Note that there is no $3\tw(G)$ pairwise vertex-disjoint paths connecting $C_i$ and $C_j$ lying on the region $\cl(C_j)\setminus \cl(C_{i})$ for two indices $i,j$ in $[\ell]$ with $j-i\geq 3\tw(G)$.
        This is because, if such vertex-disjoint paths exist, then the paths and the cycles $C_{i},\ldots, C_j$ yield a minor model of a $2\tw(G)$-grid in $H$.
        This contradicts that the treewidth of $H$ is at most $\tw(G)$.
        Therefore, we can compute a noose $O$ on $\cl(C_\ell)\setminus \cl(C_{\ell-3\tw(G)})$ intersecting at most $3\tw(G)$ vertices of $H$ in $O(\tw(G)\cdot n)$ time by the Ford-Fulkerson algorithm~\cite{ford_fulkerson_1957} such that any radial curve connecting $B$ and $C_\ell$ on the region $\cl(C_\ell)$ must intersect a vertex on $O$.

        We let $u$ and $v$ be the two vertices on $V(J')\cap V(O)$ which decompose $J'$ into three radial subcurves such that one $\pi$ does not intersect the interior of $\cl(O)$. 
        We replace $J'\setminus \pi$ by the radial subcurve of $O$ between $u$ and $v$.
        Then the modified $J'$ is a non-contractible noose on $\Sigma'$ of complexity $O(\tw(G))$ since $O$ is contractible.
        It satisfies the \textsf{Condition 1} of \Cref{sur_lem:cut_reduction}.
    \end{proof}
   
    \subparagraph*{$B'$ is contained inside of a boundary component of $F$.}
    We assume that every boundary component of $F$ is contractible, and one boundary component $\gamma$ of $F$ has another boundary component $B'\neq B$ of $\Sigma'$ inside $\cl(\gamma)$.
    If there are multiple boundary components inside $\gamma$, then we choose $B'$ that is closest to $V(\gamma)\cup V(C_\ell)$ w.\@r.\@t.\@ the radial distance on $\cl(\gamma)$.
    
Using \Cref{sur_lem:isolated_linear}, we compute a longest tight sequence $\langle C_1'=B', \ldots, C_{\ell'}'\rangle$ of concentric cycles on $\cl(\gamma)$ that are vertex-disjoint to $C_\ell$ and $\gamma$ except the last cycle $C_{\ell'}'$ by ignoring the boundary components except $B$ and $B'$. This means $C_i'$ might contain a boundary component of $\Sigma'$ other than $B$ and $B'$ inside.
    Let $u$ be the vertex on $V(C_{\ell'}')\cap V(C_\ell)$ or $V(C_{\ell'}')\cap V(\gamma)$.
    After identifying the vertices in $\cl(C_{\ell'}')$ into a single vertex $o'$, we compute the radial curve $J$ intersecting only the vertices in $\{o,u,o'\}$.
    Analogue to \Cref{sur_claim:uncontracting}, we can translate $J$ to a radial curve $J'$ satisfying \textsf{Condition 2} of \Cref{sur_lem:cut_reduction}.    
    
        Intuitively, we can translate  $J$ as a curve $J'$ connecting $B$ and $B'$ that intersects $O(\ell+\ell')$ vertices on $H$ by uncontracting $o$ and $o'$. Refer to \Cref{sur_fig:lateral}(c).
        Here, $J'$ does not intersect any other boundary components since we assumed that $\cl(C_\ell)$ has no other boundary components except $B$, and we chose $B'$ as the closest boundary component, inside $\cl(\gamma)$, to the vertices $V(\gamma)\cup V(C_\ell)$ with respect to the radial distance.
        If either $\ell$, $\ell'$, or both are at least $\omega(\tw(G))$, then we can further modify $J'$ as follows. Then the proof completes.
\begin{itemize}
    \item 
    If $\ell$ is at least $\omega(\tw(G))$, then we can compute a noose $O$ on $\cl(C_\ell)\setminus \cl(C_{\ell-3\tw(G)})$ such that $O$ separates $V(B)$ and $V(C_{\ell})$ and intersects $O(\tw(G))$ vertices in $H$.
        The details are in the proof of \Cref{sur_claim:uncontracting}.
        We let $v$ be the farthest vertex of $V(J')\cap V(O)$ from $B$ along $J'$. 
        We replace the subcurve of $J'$ between $B$ and $v$ by the noose $O$. 
        Then the obtained radial curve $J$ decomposes $\Sigma'$ into two components satisfying the \textsf{Condition 2}. Refer to \Cref{sur_fig:lateral}(d). Here, the modified $J'$ intersects at most $O(\tw(G)+\ell')$ vertices.
    \item If $\ell'$ is at least $\omega(\tw(G))$, analogously, we can compute a noose $O'$ on $\cl(C_{\ell'}')\setminus \cl(C_{\ell'-3\tw(G)}')$ such that $O'$ separates $V(B')$ and $V(C_{\ell'}')$ while it intersects $O(\tw(G))$ non-boundary vertices of $H$.
        If $O'$ does not intersect any boundary vertex, then we can modify $J'$ to traverse $O'$ analogously to the noose $O$, see \Cref{sur_fig:lateral}(e-g). 
        The modified curve $J'$ intersects at most $O(\tw(G))$ vertices of $H$, which completes the proof.
        
        We show that $O'$ does not intersect any boundary component of $\Sigma'$.  For contradiction, we suppose that $O'$ intersects a boundary component.
        We let $\bar B$ be the closest boundary component to the vertex $v'$ along $O'$, where $v'$ is the farthest vertex of $V(J')\cap V(O')$ from $B'$ along $J'$.
        Note that $\bar B$ is neither $B$ nor $B'$ since $O'$ is a noose on $\cl(C_{\ell'}')\setminus \cl(C_{\ell'-3\tw(G)}')$.
        Then there is a radial curve of complexity $O(\tw(G))$ in $H$ from $\bar B$ to $V(C_\ell)\cup V(\gamma)$ along $J'\cup O'$ while the radial distance from $B'$ to these vertices is $\ell'\in \omega(\tw(G))$, which contradicts.
        Therefore, $O'$ does not intersect any boundary component of $\Sigma'$.
\end{itemize}

        \begin{figure}
    \centering
    \includegraphics[width=0.9
    \textwidth]{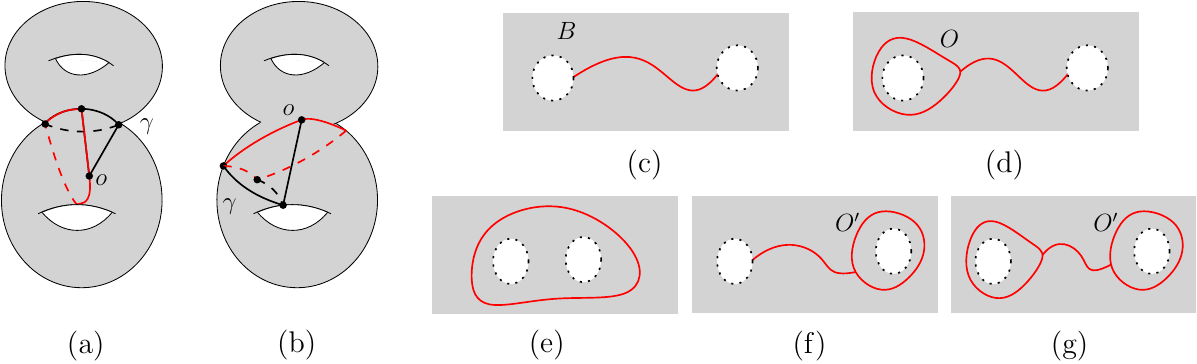}
    \caption{\small All kinds of curves obtained from \Cref{sur_lem:cut_reduction}.
    (a) Cutting $J$ reduces the Euler genus of $\Sigma'$. (b) Cutting $J$ separates the surface.
    {By replacing the handles in (a-b) with crosscaps, the figures represent non-orientable surfaces.}
    (c) Cutting $J$ reduces the number of boundary components of $\Sigma'$.
    (d-g) Among the regions enclosed by $J$, one has fewer boundary components, while the others have Euler genus zero and at least two boundary components.
    \vspace{-1em}
    }
    \label{sur_fig:lateral}
    \end{figure}

\medskip

In conclusion, in $O(\tw(G)n)$ time, we find the desired curve on $\Sigma'$ of complexity at most $O(\tw(G))$ in $H$, completing the proof of \Cref{sur_lem:cut_reduction}.
\end{proof}
In conclusion, by applying \Cref{sur_lem:cut_reduction} recursively, we may decompose $\Sigma$ into subregions of Euler genus zero, which implies \Cref{sur_thm:nice_decomp_tw}.

\SmallNiceSubgraph*
\begin{proof}
    Here, we let $b\leq k$ be the number of boundary components of $\Sigma$.
    In this proof, we show that by applying \Cref{sur_lem:cut_reduction} recursively, we can cut $\Sigma$ into at most $4g+2b-5$ regions for $g, b\geq 1$.
    This claim bounds the number of calls to \Cref{sur_lem:cut_reduction} at most $5g+2b$.
    Moreover, it bounds the total number of holes and boundary vertices at most $O(k+g)$ and $O((k+g)\tw(G))$, respectively.
    In the following,  
    we prove the claim inductively along the lexicographical ordering of the Euler genus $g$ and the number $b$ of boundary components.

   We let $J$ be the obtained radial curve by \Cref{sur_lem:cut_reduction}.
    Note that every region enclosed by $J$ has at least one boundary component.
    If $\Sigma$ has Euler genus one and exactly one boundary component, then the removal $J$ cannot decompose into multiple components. 
    Therefore, it obtains only one subsurface by cutting $J$ that is embedded on a topological disk with holes. Therefore, the claim holds. In the following, we suppose that $g>1$ or $b>1$.
    
    If $J$ is not separating $\Sigma$, then the cutting $J$ obtains one surface with Euler genus $g_s$  and $b_s$ boundary components such that $g_s\leq g-1$ and $b_s\leq b+2$, or $g_s=g$ and $b_s\leq b-1$.
    By the inductive assumption, this surface is decomposed into at most $4g_s+2b_s-5\leq 4g+2b-5$ subsurfaces with genus zero.
    Therefore, the claim holds.
    In the following, we suppose that $J$ is separating.
    We let $x$ and $y$ be the number of regions enclosed by $J$ with Euler genus at least one and zero, respectively.
    Furthermore, we let $\Sigma_1',\ldots, \Sigma_x'$ be the regions enclosed by $J$ with Euler genus at least one. 
    There are two cases: there is exactly one region enclosed by $J$ with non-zero Euler genus, or not.
    We consider each case one by one.
    \begin{itemize}    
        \item \textbf{Case for $x=1$:}
    Observe that $J$ satisfies the \textsf{Condition 2} of \Cref{sur_lem:cut_reduction}.
    Thus, $\Sigma_1'$ has at most $b-1$ boundary components.
    Then the subsurface $\Sigma_1'$ can be further decomposed into at most $4g+2b-7$ regions by the inductive assumption. 
    In total, we obtain at most $4g+2b-7+y\leq 4g+2b-5$ regions embedded on planes since $y\leq 2$.
    \item \textbf{Case for $x\geq 2$:}
    We consider the other case that 
    there are at least two regions $\Sigma_1',\ldots, \Sigma_x'$ enclosed by $J$ with non-zero Euler genus.
    We use $b_i'$ to denote the number of boundary components of $\Sigma_i'$ for $i\in[x]$.
    Note that the sum of all Euler genuses $\textsf{eg}(\Sigma'_1), \ldots \textsf{eg}(\Sigma_x')$ is at most $g$ by \Cref{sur_lem:separating_genus_reducing}.
    Furthermore, the sum of all $b_i'$ is at most $b+x$, refer to \Cref{sur_fig:lateral}.
    By the inductive assumption, each $\Sigma_i'$ is decomposed into $4\textsf{eg}(\Sigma'_i)+2b_i'-5$ regions for $i\geq y+1$.
    This concludes that $\Sigma$ is decomposed into at most $4g+2(b+x)-5x+y\leq 4g+2b-5$ regions since $x+y$ is at most three by \Cref{sur_lem:cut_reduction}.
    \end{itemize}
    In conclusion, our claim holds for every $g\geq 1$ and $b\geq 1$, and thus, \Cref{sur_thm:nice_decomp_tw} holds.
\end{proof}

\section{Linear-Time FPT Algorithm}\label{sur_sec:encoding_construction}

\ccheck{
Notably, cutting the graph and the surface according to \Cref{sur_thm:nice_decomp_tw} reduces the \textsc{$g$-Surface $k$-Disjoint Paths} problem to planar instances with $O((k+g)\tw)$ boundary vertices in total.
This already yields a linear FPT algorithm. Briefly, by considering all possible sets of terminal pairs among the boundary vertices and applying the \textsc{Planar Disjoint Paths} algorithm of Cho et al.~\cite{cho2023parameterized} as a black box, one obtains an algorithm running in time $2^{O((k^2+g^2)\tw^2)}n$.
In this section, we improve this bound by opening the black box.
More precisely, we develop a $(k+g+\tw)^{O(k+g)}n$-time algorithm for the \textsc{$g$-Surface $k$-Disjoint Paths} problem.
}

\subsection{Enumerating Weak Linkages}\label{sur_sec:crossing_pattern}
Note that, in \Cref{sur_sec:cut_reduction}, we decomposed $\Sigma$ and $G$ into planar graphs.
In this section, we enumerate $(k+g+\tw(G))^{O(k+g)}$ discrete homotopy classes for the \textsc{$g$-Surface $k$-Disjoint Paths} problem on $(G,\mathcal T)$ by applying the scheme illustrated by \Cref{sur_lem:summary_pdp}.

Let $\mathcal M$ be the subgraph of $\radgraph$ obtained by the union of terminals in $T$ and the radial curves used to cut the surface through \Cref{sur_thm:nice_decomp_tw}.
We let $\Sigma_1,\ldots, \Sigma_x$ be the subsurfaces of $\Sigma$ obtained via the decomposition induced by $\mathcal M$, and we let $\mathcal H=\{H_1,\dots, H_x\}$ be the subgraphs of $G$ embedded on them.
We aim to compute weak linkages in each radial completion $H_i^{\textsf{rad}}$ and concatenate them via the vertices on $\mathcal M$ by applying \Cref{sur_lem:summary_pdp}.
However, this introduces a subtle issue: Although two vertices lie on the same boundary component of $\Sigma_i$, they might appear in different subsurfaces since $\mathcal M$ is a union of multiple curves on $\Sigma$. 

Observe that for every pair of adjacent degree-two vertices on $\mathcal M$, both appear as boundary vertices in the same two subgraphs in $\mathcal H$.
This allows us to transfer the adjacency information on $\mathcal M$ to the corresponding boundary cycles of the subgraphs in $\mathcal H$.
Based on this, we define the auxiliary graph $\mathcal M^{\textsf{con}}$ (see \Cref{sur_fig:face_reduction}(a--b)): We identify each maximal path $P$ of degree-two vertices in $\mathcal M$ into a single vertex $v_P$, which we refer to as a \emph{portal}. Similarly, we define $\gcontract$ as the graph obtained by applying the same identification process to $\radgraph$.
Observe that if the vertices are identified to the same portal, they appear in the same boundary component of the same subsurface in $\Sigma_1,\ldots,\Sigma_x$. Furthermore, such vertices appear consecutively on the boundary components.
We then modify each $\Sigma_i$ by slightly cutting the boundary components $B_i$ so that every $B_i$ consists of the vertices in $\mathcal M$ that are identified to the same portal in $\mathcal M^{\textsf{con}}$, refer to \Cref{sur_fig:face_reduction}(c).
Unlike the original subregions $\Sigma_1,\ldots, \Sigma_x$, the modified $\widetilde\Sigma_1,\ldots, \widetilde \Sigma_x$ are not pairwise interior-disjoint.
However, if a weak linkage $\mathcal L$ of $H_i^{\textsf{rad}}$ intersects $\widetilde\Sigma_i\setminus \Sigma_i$, then no weak linkage of $H_i$ is discretely homotopic to $\mathcal L$.

We apply \Cref{sur_lem:summary_pdp} to each $H_i$ in $\mathcal H$ with respect to the modified $\widetilde \Sigma_i$, then we enumerate weak linkages in $H_i^{\textsf{rad}}$ lying on the original subsurface $\Sigma_i$.
For each combination of obtained weak linkages in $H_1^{\textsf{rad}},\ldots, H_x^{\textsf{rad}}$, we concatenate them at the portals to form weak linkages in $\gcontract$. Finally, by uncontracting the portals, we obtain weak linkages in $\radgraph$.
The following observation guarantees that if a $\mathcal T$-linkage exists in $G$, at least one of the constructed weak linkages in $\radgraph$ is discretely homotopic to it.  We call such a weak linkage a \emph{canonical weak $\mathcal T$-linkage}. 
\begin{restatable}{observation}{ObsCanonical}\label{sur_obs:canonical}
    If $G$ has a $\mathcal T$-linkage, then at least one of the obtained weak linkages is discretely homotopic to a feasible $\mathcal T$-linkage.
\end{restatable}

\begin{proof}
    We suppose that the \textsc{$g$-Surface $k$-Disjoint Paths} problem on $(G,\mathcal T)$ has a feasible $\mathcal T$-linkage.
    Then we let $\mathcal L$ be a $\mathcal T$-linkage in $G$.
    Note that it is a $\mathcal T$-linkage in $\radgraph$ also.
    For each $H_i\in \mathcal H$, we let $\Pi_i$ be the set of maximal subpaths in $\mathcal L$ that are in $H_i$.

    We let $\mathcal T_i$ be the set of all pairs $(s,t)$ such that $s$ and $t$ are connected by a path in $\Pi_i$. Then $\Pi_i$ is a $\mathcal T_i$-linkage in $H_i$.
    Additionally, let $\widetilde {\mathcal T_i}$ be the pairs of boundary components of $\widetilde 
    \Sigma_i$ corresponding to $\mathcal T_i$.
    Thus, there is a $\mathcal T_i$-linkage $\Pi_i'$ in $H_i$ and a weak $\widetilde{\mathcal T_i}$-linkage $\mathcal W_i$ in $\widetilde{H_i^{\textsf{rad}}}$
    that are discretely homotopic in $\widetilde{H_i^{\textsf{rad}}}$, where $\mathcal W_i$ is enumerated by \Cref{sur_lem:summary_pdp}.
    Here, $\widetilde{H_i^{\textsf{rad}}}$ is the graph obtained from the radial completion $H_i^{\textsf{rad}}$ by identifying each boundary component as an artificial vertex, that is a portal in $\mathcal M^{\textsf{con}}$ (and $\gcontract$).
    By replacing the subpaths in $\mathcal L$ traversing $\Pi_i$ with the paths in $\Pi_i'$, we can obtain another $\mathcal T$-linkage $\mathcal L'$. Analogously, by replacing the subpaths with the walks in $\mathcal W_i$, and then uncontracting the portals, we also provide the weak $\mathcal T$-linkage $\mathcal W$ in $\radgraph$. Here, our algorithm enumerates the weak linkage $\mathcal W$.
    This completes the proof since $\mathcal W$ is discretely homotopic to $\mathcal L$.    
\end{proof}

\begin{figure}
    \centering
    \includegraphics[width=0.8
    \textwidth]{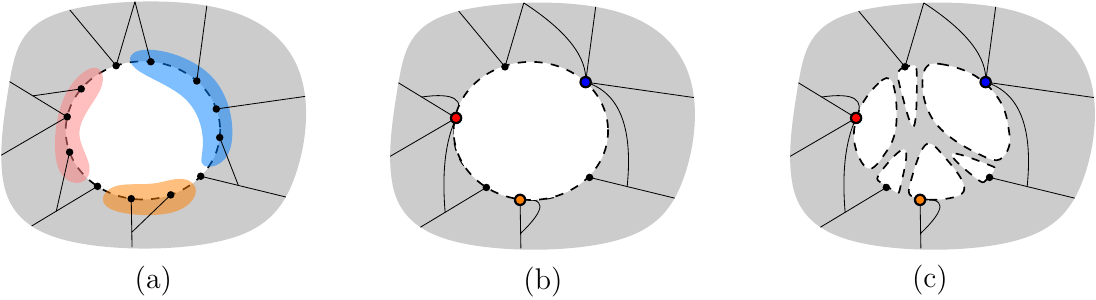}
    \caption{\small 
    (a-b) The colored areas denote maximal paths of degree-two vertices in $\mathcal M$. The vertices in the same area are identified with the same portal in $\mathcal M^{\textsf{con}}$ and $\gcontract$. 
    (c)  $\widetilde\Sigma_i$ is modified from $\Sigma_i$ so that each boundary component is incident to one boundary vertex (or a portal).
    If a weak linkage of $H_i^{\textsf{rad}}$ intersects $\widetilde\Sigma_i\setminus \Sigma_i$, then no weak linkage of $H_i$ is discretely homotopic to it.
\vspace{-1em}
}
    \label{sur_fig:face_reduction}
\end{figure}

\subsection{Reconstruction of Linkages from Canonical Weak Linkages}\label{sur_sec:completes}

\textsc{Cohomology Feasibility} is a standard technique in the study of \textsc{Disjoint Paths} problem, introduced by Schrijver~\cite{schrijver1994finding}.
For plane-embedded graphs, this technique was further developed using frames, skeleton forests, and reference paths to handle the weak linkages obtained by \Cref{sur_lem:summary_pdp} in $O(\textsf{poly}(\tw(G), k)\cdot n)$ time~\cite{cho2023parameterized}.
In this section, we extend this algorithm to surface-embedded graphs for recovering a $\mathcal T$-linkage from a canonical weak $\mathcal T$-linkage.
For this, we use our decomposition structure $\mathcal M$ defined in \Cref{sur_sec:crossing_pattern}.
The resulting algorithm runs in $O(\textsf{poly}(\tw(G), k, g)\cdot n)$ time.

We suppose that we are given frames, skeleton forests, and reference paths along with a canonical weak $\mathcal T$-linkage $\mathcal W$ on $\radgraph$.
Then we describe how to compute a $\mathcal T$-linkage in $G$, that is discretely homotopic to $\mathcal W$, in $O(\textsf{poly}(\tw(G),k,g)\cdot n)$ by applying the \textsc{Cohomology Feasibility Problem}.
Note that Schrijver~\cite{schrijver1994finding} developed a polynomial-time algorithm for the cohomology feasibility problem as follows.
First, we summarize the necessary preliminaries.

\subsubsection{Cohomology Feasibility Problem}\label{sur_sec:feasibility}

A digraph $D$ is said to be weakly connected if it is connected when the direction is ignored.
A \emph{flow function} is a function $\phi\colon E(D)\to [k]\cup [k]^{-1}\cup \{\epsilon\}$, where $\epsilon$ denotes the empty string and $[k]=\{1,\ldots, k\}$ and $[k]^{-1}=\{1^{-1},\ldots, k^{-1}\}$.
Here, we use $[k]^*$ to denote the set of all strings consisting of the symbols in $[k]\cup [k]^{-1}$.
The \emph{product} of strings $w$ and $w'$, denoted by $w\cdot w'$, is defined as the string obtained by concatenating them and deleting all appearances of  $xx^{-1}=x^{-1}x=\epsilon$.
  We say two flow functions $\phi$ and $\psi$ are \emph{$r$-homologous} for a vertex $r\in V(D)$, called the \emph{base}, if a homology function $f\colon V(D)\to [k]^*$ with $f(r)=\epsilon$ satisfies:
  
  \begin{itemize}
      \item {$f(v_e)^{-1}\cdot \phi(e)\cdot f(u_e)=\psi(e)$ for every directed edge $e=(v_e,u_e)$ in $E(D)$.}
  \end{itemize}

In the \textsc{Cohomology Feasibility Problem}, we are given a digraph $D$, a flow function $\phi$, a set $\Pi$ of \emph{undirected paths} in $D$, and a candidate function $\Gamma\colon {\Pi} \to 2^{[k]^*}$ such that every $\Gamma(\pi)$ is \emph{hereditary}, that is,
  for every string $x\in \Gamma(\pi)$, all its prefixes and $x^{-1}$ also belong to $\Gamma(\pi)$.
  Here, $\Pi$ is a set of \emph{undirected paths} in $D$. That means $\pi\in \Pi$ might use a reverse edge of $D$.
  The goal is to compute a flow function $\psi$ that is $r$-homologous to the given flow function $\phi$ such that $\psi(\pi)\in \Gamma(\pi)$ for every $\pi\in \Pi$.
  For the flow functions $\phi$ and $\psi$, $\phi(\pi)$ (and $\psi(\pi)$) denotes the product of $\phi(e)^{\mu(e)}$'s (and $\psi(e)^{\mu(e)}$'s) in order for all edges $e$ in $\pi$, where the sign $\mu(e)$ is positive if $e$ is an edge of $D$, and negative, otherwise.
  Schrijver provided a polynomial-time algorithm for the \textsc{Cohomology Feasibility Problem} for a general digraph~\cite[Section 2]{schrijver1994finding}, as summarized in \Cref{sur_lem:summarized_schrijver}. 

  \begin{restatable}[\cite{schrijver1994finding}]{lemma}{LemSummarySchrijver}\label{sur_lem:summarized_schrijver}
  For $D=(V,A)$, let $T_{\textsf{pre}}$ be the running time of computing the smallest pre-feasible function.
      The cohomology feasibility problem can be solved in $O(\chi T_{\textsf{pre}}+\chi^2 |V|)$ time, where $\chi$ denotes the number of $\pi$ in $\Pi$ with $\phi(\pi)\notin \Gamma(\pi)$.
\end{restatable}

Here, a function $f: V\to [k]^*$ with $f(r)=\epsilon$ 
  is called a \emph{pre-feasible (homology) function} if 
  either $\psi(\pi)\in \Gamma(\pi)$, or $f(v)=f(u)=\epsilon$ for each path $\pi$ of $\Pi$, where $v$ and $u$ are the starting and end vertices of $\pi$. 
  Here, $r$ is the base vertex in $V$ of the homologous, which means our goal is to find the $r$-homologous flow function to $\phi$ in $D$.
  For two functions $f$ and $g:V\to [k]^*$, we say $f$ is \emph{smaller} than $g$ if $f(v)$ is a prefix of $g(v)$ for every $v\in V$.
  Schrijver showed that
  the \emph{smallest} pre-feasible function $\overline f$ larger than $f$ is well-defined for any function $f$ in a directed graph $D$. That is, any finite pre-feasible function $f^*$ larger than $f$ is also larger than $\overline f$. 
  \begin{proof}[Sketch of \Cref{sur_lem:summarized_schrijver}]
When a weak connected directed graph $D=(V,A)$, a flow function $\phi: A\to [k]\cup[k]^{-1}\cup \{\epsilon\}$, a set $\Pi$ of undirected paths in $D$, and the candidate set $\Gamma: \Pi\to [k]^{*}$ are given,
Schrijver~\cite[Section 2]{schrijver1994finding} designed a polynomial-time algorithm solving the cohomology feasibility problem as illustrated by \Cref{sur_lem:summarized_schrijver}.
    We sketch his algorithm.

    For given $\Pi$, $\Gamma$, and $\phi$, the goal is to compute a homology function $f$ such that $\psi(\pi)\in\Gamma(\pi)$ for all paths $\pi\in\Pi$, where $\psi(e)=f(v_e)^{-1}\cdot \phi(e) \cdot f(u_e)$ for each directed edge $e=(v_e,u_e)$ in $D$.
  In this case, we say the function $f$ is \emph{feasible}.
  By the definition, a feasible function is larger than any smallest pre-feasible function addressing all the paths $\pi\in \Pi$ with $\phi(\pi)\notin \Gamma(\pi)$. This motivates his polynomial-time algorithm by applying the \emph{join} operation.

  We define the \emph{join} operation for two strings (or functions) that returns the smallest string (or function) larger than both of the two inputs.
    Particularly, for two strings $x$ and $y$ in $[k]^*$, we define the \emph{join} of them, denoted by $x\vee y$, so that $x\vee y=y\vee x=y$ if $x$ is a prefix of $y$, otherwise, we set $x\vee y$ is infinite. Furthermore, for two functions $f$ and $g: V\to [k]^*$,  the join $f\vee g$ is defined as $(f\vee g)(v)=f(F)\vee g(v)$ for all vertices $v$. We say $f\vee g$ is finite if all $(f\vee g)(v)$ is finite, otherwise, we say it is infinite. Note that if $f \vee g$ is finite, then $(f\vee g)(v)$ is equal to either $f(v)$ or $g(v)$ for all $v\in V$. 
    Joining two functions takes $O(|V|)$ time.
  
  Schrijver's algorithm~\cite{schrijver1994finding} first enumerates all $f_{\pi}$ and $g_{\pi}$ for $\pi\in \Pi$ with $\psi(\pi)\notin \Gamma(\pi)$, where $f_{\pi}$ (and $g_{\pi}$) maps all vertices in $V$ to $\epsilon$ except $f_{\pi}(v)=\phi(e)$ (and $g_{\pi}(u)=\phi(e')^{-1}$). Here, $v$ and $u$ (and $e$ and $e'$) are the starting and end vertices (and edges) of $\pi$.
  Observe that any function $f$ addressing the violation of the path $\pi$ is larger than either $f_{\pi}$ or $g_{\pi}$.
  Then the algorithm computes all smallest pre-feasible functions $\overline f_{\pi}$'s and $\overline g_{\pi}$'s. 
  By \emph{joining} either $\overline f_{\pi}$ or $\overline g_{\pi}$ for all $\pi$, the algorithm returns a finite feasible function if it exists.
  
    Note that the join phase takes $O(\chi^{2}|V|)$ time by applying the 2-SAT algorithm~\cite{4567876} after checking the join for every pairs of $\overline f_{\pi}$'s and $\overline g_{\pi}$'s, where $\chi$ denotes the number of $\pi$ in $\Pi$ with $\phi(\pi)\notin \Gamma(\pi)$. 
    Therefore, the algorithm runs in $O(\chi\cdot T_{\textsf{pre}}+\chi^2 |V|)$ time. 
  \end{proof}

  \subparagraph*{Reduction for the Disjoint Paths problem.}
  When a collection $\mathcal L$ of pairwise edge-disjoint and non-crossing walks is given on a surface-embedded graph $H$, we can find a linkage $\mathcal L'$ that is discretely homotopic to $\mathcal L$, if it exists, by reducing the problem to the \textsc{Cohomology Feasibility Problem}.
  For this, we first translate $H$ into a digraph by replacing its edges with two directed edges in opposite directions.
  Then we define its dual digraph $D=(\mathcal F, A)$, where $\mathcal F$ is the set of faces of $H$.
  Here, an edge $e$ in $H$ corresponds to a dual edge $e^*=(L_e,R_e)$, where $L_e$ (and $R_e$) is the left (and right) face of $e$. 
  {To define $R$-homology on the dual digraph $D$, we first designate a specific face $R$ in $H$ to serve as the base for homology. To this end, we insert a dummy directed loop at an arbitrary end vertex of $\mathcal L$ in $H$, which forms a contractible cycle on $\Sigma$ enclosing a small disk $R$. This modification introduces a new face $R$ in $H$. We then define and use $R$-homologous on the dual digraph $D$.} 
  
  For the walks in $\mathcal L$ on $H$, we define a flow function $\phi\colon A \to [k]\cup[k]^{-1}\cup \{\epsilon\}$, where $k=|\mathcal L|$, such that for a dual edge $e^*$ in $D$, $\phi(e^*)=i$, if the edge $e$ in $H$ is used by the $i$th walk in $\mathcal L$, and otherwise, $\phi(e^*)=\epsilon$.
    To complete the reduction, we define the path set $\Pi$ and its candidates $\Gamma$.
    The following is a sufficient condition for $\Pi$ and $\Gamma$ so that the feasible flow function $\psi$ represents a pairwise vertex-disjoint paths connecting the same pairs of endpoints as the walks in $\mathcal{L}$:
    \begin{itemize}
        \item [\textsf{(i)}] $\Pi$ includes all undirected paths $\pi$ in $D$ such that the duals of edges in $\pi$ are incident to a common vertex $v_\pi$ in $H$,
        \item [\textsf{(ii)}] $\Gamma(\pi)\subseteq[k]\cup[k]^{-1}\cup\{\epsilon \}$ for the paths $\pi$ in from \textsf{Condition (i)}, and
        \item [\textsf{(ii)}] $\Gamma(\pi)=\{i,i^{-1},\epsilon \}$ for the paths $\pi$ in  \textsf{Condition (i)} if the common vertex $v_{\pi}$ in $H$ is an end vertex of the $i$th walk in $\mathcal L$.
    \end{itemize}
    
    The solution for the \textsc{Cohomology Feasibility Problem} for $\phi$, $\Pi$, and $\Gamma$ as above derives pairwise vertex-disjoint paths. Note that the homologous relations are closed under the discrete homotopy relation.
    Particularly, for two weak linkages $\mathcal W$ and $\mathcal W'$, if $\mathcal W'$ is obtained by sequential face operations from $\mathcal W$, then they correspond to homologous flow functions as stated in the lemma below due to Lokshtanov et al.~\cite{lokshtanov2020exponential}. 
    Although the authors established the observation for planar graphs, it extends naturally to surface-embedded graphs, as the face operation is defined locally on the surface.   The correctness of the above reduction is guaranteed by \Cref{sur_lem:homology-homotopy}.

    \begin{lemma}[{\cite[Lemma 5.2]{lokshtanov2020exponential}}]\label{sur_lem:homology-homotopy}
    Let $\mathcal W$ and $\mathcal W'$ be two weak linkages discretely homotopic to each other. 
    Then their flow functions are $R$-homologous to each other. 
    \end{lemma}

\subparagraph*{Overview of remaining \Cref{sur_sec:completes}.}
To compute a $\mathcal T$-linkage in $G$ that is discretely homotopic to $\mathcal W$ on $\radgraph$, there are three phases: modifying the canonical weak $T$-linkage $\mathcal W$ on $\radgraph$ as pairwise edge-disjoint walks $\wmod$, reducing the problem to a cohomology feasibility problem, and developing the $\textsf{poly}(\tw(G), k, g)\cdot n$ time algorithm for the cohomology feasibility problem for $\wmod$.
According to \Cref{sur_lem:summarized_schrijver}, there are two tasks to achieve a $\textsf{poly}(\tw(G), k, g)\cdot n$-time algorithm for the cohomology feasibility problem for $\wmod$: bounding the number $\chi$ of paths in $\Pi$ violating $\Gamma$ during the first two phases, and achieving the smallest pre-feasible function in $\textsf{poly}(\tw(G), k, g)\cdot O(n)$ time in the last phase.

Here, our given frames, skeleton forests, and reference paths are constructed with respect to $H_i$'s in $\mathcal H$ and their boundary vertices and terminals in $\mathcal M$.
Note that the frames decompose the surface $\Sigma$ into $O(k+g)$ regions so that half of them do not contain any vertices on $\mathcal M$ inside.
Recall that such a region is enclosed by two homeomorphic frames, and we call such a region a \emph{boundary-free annulus}.
The canonical weak $\mathcal T$-linkage $\mathcal W$ satisfies the following:
\begin{itemize}
    \item $\mathcal W$ is set of pairwise vertex-disjoint walks in $G$ except on frames, skeleton forests, or $\mathcal M$, 
    \item For a boundary-free annulus $\circledcirc$, $\mathcal W$ traverses in $\circledcirc$ along the reference paths, and
    \item If $\mathcal W$ traverses an edge in $G^{\textsf{rad}}$ but not in $G$, then it is an edge on frames, skeleton forests, or $\mathcal M$.
        Additionally, such an edge is traversed at most $O(\tw(G))$ times.
\end{itemize}

\subsubsection{Modification of $\mathcal M$ and $\radgraph$ for $\wmod$ and $\gmod$}
It is easy to translate a weak linkage as pairwise edge-disjoint walks by duplicating edges in $\radgraph$ which are traversed by $\mathcal W$ more than once.
However, to improve the running time in the following phases, we require further modification.

We use $\gparallel$ to denote the obtained graph from $\radgraph$ by adding the parallel edges. 
Then the walks in $\mathcal W$ are pairwise edge-disjoint on $\gparallel$.
We say an edge $e$ in $\gparallel$ is a \emph{linkage-edge} if a walk of $\mathcal W$ uses $e$. 
 We say an edge $e$ lies \emph{between $e_1$ and $e_2$} if $e_1$, $e$, and $e_2$ are incident to a common vertex $v$, and they appear in a clockwise direction around $v$. 
 Two linkage-edges $e$ and $e'$ incident to $v$ form a \emph{wedge at $v$} if there is no other linkage-edge between $e$ and $e'$.
 All wedges at $v$ are pairwise interior-disjoint. 
 For a vertex $v$, the number of wedges is the same as the number of linkage-edges incident to $v$.
 We say a wedge at $v$ is \emph{empty} if 
 no edge of $G$ incident to $v$ is contained in the wedge. 
 
For each non-empty wedge at a vertex $v$ in the frames, skeleton forests, and $\mathcal M$, we insert a new vertex $v'$ to $\gparallel$, 
 and add the edge between $v$ and $v'$. 
 Then we remove the edges in the wedge and reconnect them to $v'$ instead of $v$. Note that, the edges incident to a new vertex $v'$ are not linkage-edges. 
 We do this for all non-empty wedges and all vertices on the frames, skeleton forests, and $\mathcal M$.
 This takes $O(n)$ time in total.
 Observe that a one-to-one correspondence exists between the weak linkages in $\gparallel$ and in $\gmod$.
 We let $\mathcal W_{\textsf{mod}}$ be the weak linkage in $\gmod$ corresponding to the canonical weak linkage $\mathcal W$ in $\gparallel$.

Observe that $\wmod$ traverses a vertex in $\gmod$ more than once, then the vertex lies on the frames, skeleton forests, or $\mathcal M$. Additionally, such a vertex has a degree at most $O(\tw(G))$ on $\gmod$. 
It is crucial to design an efficient algorithm for the cohomology feasibility problem.
Briefly, it bounds the number $\chi$ of violating paths in \Cref{sur_lem:summarized_schrijver} by \Cref{sur_lem:small_violating}. 
Note that the treewidth of $\gparallel$ and $\gmod$ are both at most $O(\tw(G))$ by construction. 
For clarity, we use the fact instead of the explicit $\tw(\gparallel)$ or $\tw(\gmod)$.


\subsubsection{Reduction to Cohomology Feasibility Problem}
Let $D=(\mathcal F, A)$ be the directed dual graph of $\gmod$ embedded on $\Sigma$ constructed by the reduction in \Cref{sur_sec:feasibility}.
Recall that there is a dummy face $R\in \mathcal F$ that corresponds to a loop at an arbitrary terminal in $T$.
We refer to the $R$-homologous relation as the homologous relation, briefly.

We let $\phi$ be the flow function mapping $A\to [k]\cup[k]^{-1}\cup \{\epsilon\}$ on $D$ corresponding to $\wmod$ in $\gmod$, where $k=|\mathcal T|$.
Furthermore, the path set $\Pi$ and its candidates $\Gamma$ are defined by the \textsf{Condition(i-iii)} in \Cref{sur_sec:feasibility}.
Then the solution of the cohomology feasibility problem defined on $D$ with respect to $\phi,\Pi$, and $\Gamma$ 
corresponds to $\mathcal T$-linkage in $\gmod$ that is discretely homotopic to $\wmod$.
The additional condition ensures correspondence to a $\mathcal T$-linkage in $G$:
\begin{enumerate}
    \item [\textsf{(iv)}] $\Gamma(\pi)=\{\epsilon\}$ for the paths $\pi$ in the \textsf{Condition (i)}, where the dual of edges in $\pi$ is incident to a common vertex $v_\pi$ in $\gmod$, such that $v_\pi$ is not a vertex in $G$.
\end{enumerate}

To improve the running time in the following phase, we require further modification.
Particularly, we insert paths into $\Pi$ with respect to the frames and paths on $\mathcal M$.
For each frame $C$, we insert a path $\pi_C$ into $\Pi$ consisting of the dual of all edges in $\gmod$ incident to the frame $C$ in the boundary-free annulus enclosed by $C$.
Then we set $\Gamma(\pi_C)$ as the set of prefixes of $\phi(\pi_C)$ and their reverse.
Note that for each maximal path $\gamma$ in $\mathcal M$ identified into the same vertex for $\mathcal M^{\textsf{con}}$, defined in \Cref{sur_sec:crossing_pattern}, the incident edges and faces to $\gamma$ compose two disjoint face-edge paths in $\gmod$.
We insert their corresponding paths $\pi_{\gamma}$ and $\pi_{\gamma}'$ in $D$ into $\Pi$, and set $\Gamma(\pi_{\gamma})$ (and $\Gamma(\pi_{\gamma}')$) as the set of prefixes of $\phi(\pi_{\gamma})$ (and $\phi(\pi_{\gamma}')$) and their reverse.

Note that there is a discretely homotopic $\mathcal T$-linkage in $G$ whose intersection pattern with the frames and the radial curves $\gamma$ on $\mathcal M$ matches that of the canonical weak linkage $\mathcal W$ (and $\wmod$). Please refer to~\cite{cho2023parameterized} and the proof of \Cref{sur_obs:canonical}. 
Therefore, the correctness of this reduction is guaranteed by \Cref{sur_lem:homology-homotopy}.
Furthermore, the number of paths $\pi$ in $\Pi$ with $\phi(\pi)\notin \Gamma(\pi)$ is bounded by \Cref{sur_lem:small_violating}.

\begin{lemma}\label{sur_lem:small_violating}
 There are at $O(\tw(G)^3(k+g))$ paths $\pi$ in $\Pi$ with $\phi(\pi)\notin \Gamma(\pi)$.
\end{lemma}
\begin{proof}
By construction, the paths in $\Pi$ corresponding to a frame or a path in $\mathcal M$ satisfies $\Gamma$. 
That means a path $\pi\in \Pi$ with $\pi\notin \Gamma(\pi)$ satisfies that the dual of all edges in $\pi$ are incident to a common vertex $v_\pi$ in $\gmod$.
Additionally, the common vertex $v_\pi$ is traversed by $\wmod$ more than once, or $v_\pi$ is a vertex in $\gmod$ but not in $G$.
Such a vertex lies on frames, skeleton forest, and $\mathcal M$. 
Note that the complexity of the structures is at most $O(\tw(G)\cdot(k+g))$.
Furthermore, each of such a vertex has degree $O(\tw(G))$ in $\gmod$ by construction.
Therefore, there are $O(\tw(G)^3\cdot(k+g))$ paths that do not satisfy $\Gamma$.
\end{proof}

\subsubsection{Smallest Pre-feasible Function}

For computing the smallest pre-feasible function $\overline f$ of a given $f:\mathcal F \to [k]^{*}$, we make updates to $f(\cdot)$ at each iteration.
Each iteration selects a face $F$ and updates $f(F)$ to handle some path $\pi\in \Pi$ with $\phi(\pi)\notin \Gamma(\pi)$.
This process increases the length of $f(F)$.
Cho et al.~\cite{cho2023parameterized} described a data structure that supports the updating process in $O(k)$ time.
Moreover, the number of updating iterations is at most $L\cdot |\mathcal F|$, where $L$ is the longest length of $\overline f(F)$ for $F\in \mathcal F$.
The following lemma bounds $L$ as $O(\tw(G)(k+g))$.

\begin{lemma}\label{sur_lem:length}
The length of $\overline f (F)$ is $O(\tw(G)(k+g)^2)$ or infinite for a face $F\in \mathcal F$. 
\end{lemma}
\begin{proof}

Cho et al. proved the align lemma~\cite[Lemma 8.4]{cho2023parameterized} for planar graphs by establishing the claim that
for a boundary-containing region and boundary-free annulus, if there is a frame or a boundary component $C$ incident to the region such that the length of $\overline f(F')$ is at most $L_C$ for every face $F'$ incident to $C$, then the length of $\overline f(F)$ is at most $L_C+O(M)$ for every face $F$ in the region. 
Here, $M$ denotes the maximum complexity of each frame or boundary component. 
The claim also extends to our regions of the surface $\Sigma$ enclosed by the frames and $\mathcal M$ since they are topological disks with holes.
In our setting, we can set $M=O((k+g)\tw(G))$ by \Cref{sur_thm:nice_decomp_tw}.

Note that $\overline f(R)=\epsilon$, where $R$ is the face enclosed by the dummy loop at some terminal in $T$.
Our goal is to achieve $R$-homologous flow $\psi$ to $\phi$.
Since the number of framed regions are $O(k+g)$ in the surface $\Sigma$ by \Cref{sur_thm:nice_decomp_tw},
$\overline{f}(F)$ has length $O(\tw(G)(k+g))$ for all faces $F$ of $\gmod$ by inductively applying the above claim.
\end{proof}

\medskip

In conclusion, we can compute the smallest-prefeasible function $\overline f$ for any function $f$ in $O(\tw(G)(k+g)^3\cdot n)$ time. 
Along with \Cref{sur_lem:small_violating}, we can find a $\mathcal T$-linkage discretely homotopic to a weak linkage $\mathcal W$ enumerated by \Cref{sur_sec:crossing_pattern} in $O(\textsf{poly}(\tw(G),k,g) \cdot n)$ time if exists according to \Cref{sur_lem:summarized_schrijver}.

\subsection{Conclusion}
In conclusion, after the $O((k+g)\tw(G)n)$-time preprocessing from \Cref{sur_thm:nice_decomp_tw}, we enumerate $(k+g+\tw(G))^{O(k+g)}$ weak linkages.
Thus, we can solve the \textsc{$g$-Surface $k$-Disjoint Paths} problem in $(k+g+\tw(G))^{O(k+g)}n$ time.
\ccheck{Furthermore, the treewidth can be reduced to $2^{O(k+g)}$ in $2^{O(k+g)}n$ time by applying the irrelevant vertex technique in \Cref{sec:irrelevant}. Plugging this bound into the algorithm above yields a $2^{O(k^2+g^2)}n$-time algorithm.}
Along with \Cref{sur_cor:kernelization}, this concludes \Cref{sur_thm:main_result}.

\section{Irrelevant Vertex Technique}\label{sec:irrelevant}
For the \textsf{$g$-Surface $k$-Disjoint Paths} problem on $(G,\mathcal T)$, it is well known that there exists a computable function $f(k, g)$ such that any $f(k, g)$-isolated vertex is irrelevant~\cite{Geelen2018Explicit}. Recall that a vertex $v$ is \emph{irrelevant} if $(G, \mathcal T)$ is a \textsf{YES}-instance if and only if $(G-v, \mathcal T)$ is a \textsf{YES}-instance.
Furthermore, we say a contractible cycle (or noose) $C$ \emph{isolates} a vertex $v$ if its enclosing region $\cl(C)$ contains $v$ but excludes all terminals $T$.
Additionally, we say a vertex $v$ is $\ell$-isolated if there is a sequence $\langle C_1,\ldots, C_{\ell}\rangle$ of \emph{concentric} cycles isolating $v$. Here, $C_1,\ldots, C_\ell$ are said to be concentric if they are pairwise vertex-disjoint contractible cycles with $\cl(C_{i})\subsetneq \cl(C_{i+1})$ for $i\in[\ell-1]$.
In this section, we describe how to remove irrelevant vertices so that the reduced graph has a bounded treewidth in $2^{O(k+g)}n$ time.
Note that for the plane-embedded graphs, Cho et al.~\cite{cho2023parameterized} designed such an algorithm as derived by \Cref{sur_thm:irr_plane}.
\begin{theorem}[{\cite[Theorem 4.1]{cho2023parameterized}}]\label{sur_thm:irr_plane}
Given a graph $G$ embedded on a plane with $b$ boundary components, we can remove $\ell$-isolated vertices from $G$ in $2^{O(b)}|V(G)|$-time in total so that the resulting graph $G$ has treewidth of $2^{O(b)}\cdot \ell$ for any positive integer $\ell$.
\end{theorem}

Note that Golovach et al.~\cite{golovach2025irrelevant} designed an algorithm that removes vertices $v$ if there is a sequence $\langle C_1,\ldots, C_{\ell}\rangle$ of \emph{concentric} cycles isolating $v$, and there are $\ell$ vertex-disjoint paths connecting $C_1$ and $C_\ell$, where $\ell$ is a sufficiently large constant computable in the Euler genus $g$ and the number $k$ of terminal pairs $\mathcal T$. Note that such vertices are irrelevant on $(G,\mathcal T)$.
For this, they used Courcelle's theorem, and thus, their algorithm has a large complexity in $g$ and $k$.
\ccheck{However, Mazoit~\cite{mazoit2013single} claimed an explicit function $f(k,g)=2^{O(k+g)}$ so that $f(k,g)$-isolated vertex is an irrelevant vertex on $(G,\mathcal T)$, and recently it was strictly reproved by Cavallaro et al.~\cite[Theorem 4.10]{cavallaro2026optimalboundskdisjointpaths}.}
Therefore, we can modify the prior algorithm to run more efficiently.
Independently, we can also achieve such a $2^{O(k+g)}n$-time algorithm by applying our surface cut reduction technique, illustrated by \Cref{sur_lem:cut_reduction}.
Briefly, we decompose the surface-embedded graph into several plane-embedded graphs while removing some $f(k,g)$-isolated vertices. 
\ccheck{Here, we suppose that $f(k,g)$ is the explicit function $2^{O(k+g)}$~\cite{cavallaro2026optimalboundskdisjointpaths,mazoit2013single} so that $f(k,g)$-isolated vertex is an irrelevant vertex on $(G,\mathcal T)$.}
Then, we can further remove isolated vertices in each of the plane-embedded graphs by applying the algorithm developed by Cho et al.~\cite{cho2023parameterized}.
Then the obtained surface-embedded graph has a bounded treewidth as as illustrated by \Cref{thm:irrelevant_surface}. Details are in the proof.

\begin{theorem}\label{thm:irrelevant_surface}
    For the \textsf{$g$-Surface $k$-Disjoint Paths} problem, we can reduce the treewidth by $2^{O(k+g)}$ by removing $f(k,g)$-isolated vertices in $2^{O(k+g)}n$ time.
\end{theorem}
\begin{proof}
We first modify the algorithm of \Cref{sur_lem:cut_reduction} so that the complexity of cutting curves is bounded by $f(k,g)=2^{O(k+g)}$ instead of the treewidth. The following claim derives the modified algorithm.

\begin{claim}\label{sur_lem:irrelevant_cut_reduction}
    Let $H$ be the subgraph of $G$ embedded on a subsurface $\Sigma'$ of $\Sigma$ with genus $g'\leq g$ and $b'$ boundary components. If $g',b'\geq 1$, then we can remove $f(k,g)$-isolated vertices, and compute the radial curve $J$ on $\Sigma'$, in $2^{O(b')}\cdot f(k,g)\cdot  n$ time, of complexity $2^{O(b')}\cdot f(k,g)$ such that $J$ encloses at most three regions of $\Sigma'$, and one of the following holds:
    \begin{enumerate}
        \item Every region of $\Sigma'$ enclosed by $J$ has Euler genus at most $g'-1$ or
        \item A region of $\Sigma'$ enclosed by $J$ either \textsf{(i)} has fewer than $b'$ boundary components and Euler genus $g'$, or \textsf{(ii)} has at least two boundary components and Euler genus zero.
    \end{enumerate}
\end{claim}

\begin{proof}
    The algorithm is the same as the surface cut reduction algorithm illustrated in \Cref{sur_lem:cut_reduction}.
Briefly, we find one tight concentric cycles $\langle C_1,\dots, C_\ell\rangle$ isolating a boundary component $B$ of $\Sigma'$ from the other boundary components so that there is a non-contractible noose $J$ of complexity at most three satisfying the \textsf{Condition 1} of \Cref{sur_lem:irrelevant_cut_reduction} when we identify the vertices on $\cl(C_\ell)$ as an artificial vertex $o$.
Otherwise, we find such tight concentric cycles  $\langle C_1,\dots, C_\ell\rangle$ isolating $B$ and another tight concentric cycles $\langle C_1',\dots, C_{\ell'}'\rangle$ so that (1) $\cl(C_1')$ includes another boundary component $B'\neq B$, and (2)  the radial distance between $C_\ell$ and $C_{\ell'}'$ is at most one.
By contracting the vertices in $\cl(C_\ell)$ and $\cl(C_{\ell'}')$ as artificial vertices $o$ and $o'$, respectively, we can find a radial curve $J$ of complex at most three and satisfy the \textsf{Condition 2} of \Cref{sur_lem:irrelevant_cut_reduction}.
This process takes $O(|E(G)|)$ time.
By uncontracting the obtained curve $J$, we can obtain the radial curve satisfying the conditions in \Cref{sur_lem:irrelevant_cut_reduction} except for the complexity. 
In the following, we show that we can reduce the complexity bounded by $f(k,g)$ by removing some isolated vertices.

Note that $\cl(C_\ell)$ (and $\cl(C_{\ell'}')$) is $c$-punctured plane with $c\leq b'+2$.
Moreover, if a vertex is $f(k,g)$-isolated in $\cl(C_\ell)$ (and $\cl(C_{\ell'}')$), then it is also $f(k,g)$-isolated on $\Sigma$.
Therefore, we can remove $f(k,g)$-isolated vertices in the regions by \Cref{sur_thm:irr_plane} so that the subgraph embedded on $\cl(C_\ell)$ (and $\cl(C_{\ell'}')$) have treewidth at most $2^{O(b')}\cdot f(k,g)$.
Then we can obtain the radial curve satisfying \Cref{sur_lem:irrelevant_cut_reduction} analogously to \Cref{sur_claim:uncontracting}.
This returns a radial curve satisfying \Cref{sur_lem:irrelevant_cut_reduction} of complexity $2^{O(b')}\cdot f(k,g)$.
Here, both removing isolated vertices and uncontracting processes take $2^{O(b')}n$-time.
This completes the proof.
\end{proof}

In the following, we describe how to conclude \Cref{thm:irrelevant_surface} by applying \Cref{sur_lem:irrelevant_cut_reduction} and \Cref{sur_thm:irr_plane} recursively.
We start by adding $2k$ empty holes at each of the terminals on $\Sigma$.
According to the proof of \Cref{sur_thm:nice_decomp_tw}, by $O(k+g)$ calls of \Cref{sur_lem:irrelevant_cut_reduction}, we decompose the surface into $O(k+g)$ number planes each with $O(k+g)$ holes and $2^{O(k+g)}\cdot f(k,g)$ boundary vertices. This is because during the process, the number of boundary components $b'$ is at most $O(k+g)$.
We apply \Cref{sur_thm:irr_plane} by setting $\ell=f(k,g)$, and then, we reduce the treewidth of each subgraph embedded on a plane with boundary components by $2^{O(k+g)}\cdot f(k,g)=2^{O(k+g)}$.

Since the number of boundary vertices is $2^{O(k+g)}$, we can obtain that the treewidth of the reduced graph $G$ is at most $2^{O(k+g)}$.
This completes the proof.
\end{proof}

\section{Enumerating Weak Linkages for Planar Graphs with Holes}\label{ap:enumerating}
In this section, we present detailed algorithms corresponding to \Cref{sur_lem:summary_pdp} restated below.
Roughly, we provide an extension of the approach by Cho et al.~\cite{cho2023parameterized} to planar graphs embedded on a plane with no hole, i.e., no boundary component.
Throughout this section, we assume that $G$ is a planar graph embedded on a plane with $b$ boundary components.
We refer to the vertices lying on each hole as \emph{boundary vertices}.
We first contract each boundary component into a single vertex, and denote the resulting graph by $\widetilde G$.
Let $\widetilde T$ denote the set of contracted boundary vertices.
Note that $\widetilde G$ is embedded on a plane without holes and $|\widetilde T| = b$,  refer to~\Cref{sur_fig:framed_regions}(a).
Furthermore, for any set $\mathcal T$ of boundary vertex pairs in $G$, a weak $\mathcal T$-linkage corresponds to a weak $\widetilde{\mathcal T}$-linkage in $\widetilde G$, where $\widetilde{\mathcal T}$ is a set of pairs of $\widetilde T$ of size at most $M'$.
Here, $M'$ denotes twice the number of boundary vertices in $G$.
Precisely, if a boundary vertex $v \in \widetilde T$ is identified with $x$ boundary vertices in $G$, then it is sufficient to allow $v$ to appear in $\widetilde{\mathcal T}$ at most $2x$ times.
This correspondence allows us to directly apply the techniques from~\cite{cho2023parameterized} to $\widetilde G$.

\LemSummaryPDP*
\subsection{Constructing Structures}
Here, we describe how to construct the structures: \emph{frames}, \emph{skeleton forests}, and \emph{reference paths}, as illustrated in \Cref{sur_lem:summary_pdp}.

\subparagraph*{Frames.}
We first to decompose $\widetilde G$ and its embedding into $O(b)$ regions using $O(b)$ nooses, called \emph{frames}.
Each frame contains at most $O(\tw(G))$ vertices.
Here, $\tw(G)$ and $n$ denote the treewidth and the number of vertices of $G$, respectively.
Note that the treewidth and number of vertices of $\widetilde G$ are at most $\tw(G)$ and $n$, respectively.

We first start from an arbitrary boundary vertex $v$ in $\widetilde T$.
Then we compute a longest tight sequence $\langle I_1,\ldots, I_{\ell}\rangle$ of concentric cycles in $\widetilde G$ isolating $v$ using \Cref{sur_lem:isolated_linear}.
Let $\textsf{Ring}(I_i,I_j)$ be a region bounded by two cycles $I_i$ and $I_j$ for $0<i<j<\ell$.
We say $\textsf{Ring}(I_i,I_j)$ is a \emph{maximal boundary-free ring} if it has no boundary vertex in $\widetilde T$ inside, and
$\textsf{Ring}(I_{i'},I_{j'})$ contains at least one for any two indices $i',j'$ with 
$[i,j]\subsetneq [i',j']$.
Also, $\textsf{Ring}(I_i,I_j)$ is said to be \emph{thick} if  $|i-j|>100\tw(G)$. 

Let $\textsf{Ring}(I_i, I_j)$ be a maximal boundary-free ring that is thick.
We compute a noose $C$ (and $C'$) that corresponds to a minimum vertex cut between the vertices of $I_{i+30\tw(G)}$ and $I_{i+40\tw(G)}$ (respectively, $I_{j-40\tw(G)}$ and $I_{j-30\tw(G)}$) lying in $\ring(I_{i+30\tw(G)}, I_{i+40\tw(G)})$ (respectively, $\ring(I_{j-40\tw(G)}, I_{j-30\tw(G)})$).
We call these nooses \emph{frames}.
Note that each frame consists of $O(\tw(G))$ vertices. This is because if $\omega(\tw(G))$ vertex-disjoint paths exist between $I_{i+30\tw(G)}$ and $I_{i+40\tw(G)}$, then together with the cycles $I_{i+30\tw(G)}, \dots, I_{i+40\tw(G)}$ they would form a minor model of a $10\tw(G)$-grid, which is a contradiction.
The two frames $C$ and $C'$ within $\textsf{Ring}(I_i, I_j)$ enclose a subregion of the ring.
We refer to this frame-bounded subregion as a \emph{boundary-free annulus}.

We recursively move to the next boundary vertex $v' \neq v$ in $\widetilde T$ by removing the vertices in $I_1,\ldots, I_{\ell}$.
We compute a longest tight sequence $\langle I'_1,\ldots, I'_{\ell'}\rangle$ of concentric cycles in $\widetilde G$ after deleting the vertices in $I_1,\ldots, I_{\ell}$. Next, we compute the frames with respect to their thick maximal boundary-free rings, and then move to the other boundary vertex in $\widetilde T\setminus\{v,v'\}$.
The recursive algorithm takes $O(b \cdot \tw(G) \cdot n)$ time, and it gives us $O(b)$ frames each of size $O(\tw(G))$.
Observe that the \emph{tightness} of the cycle sequences ensures that the frames (and the boundary-free annuli defined by them) are pairwise disjoint.

\subparagraph*{Boundary-free annuli and boundary-containing regions.}
The frames decompose the plane into $O(b)$ regions.
Half of them are boundary-free annuli, where each annulus is bounded by two frames that are separated by at least $\Omega(\tw(G))$ in radial distance.
We refer to the remaining regions as \emph{boundary-containing regions}.
Unlike boundary-free annuli, a boundary-containing region has either a boundary vertex in $\widetilde T$ or at least three bounding frames.
Importantly, in a boundary-containing region, any boundary vertex from $\widetilde T$ within the region and the bounding frames of the region are separated from each other by at most $O(\tw(G))$ in radial distance.
In the following, we describe the substructures, called the \emph{skeleton forest} and the \emph{reference paths}, associated with each boundary-containing region and each boundary-free annulus.

\subparagraph*{Skeleton forest and reference paths.}
It is clear that each connected boundary-containing region $R^{\textsf o}$ admits $O(b')$ radial curves, each of complexity $O(\tw(G))$, whose cutting transforms $R$ into a topological disk $\Delta$.
Furthermore, all boundary vertices of $R^{\textsf o}$ (from $\widetilde T$) and the bounding frames of $R^{\textsf o}$ appear on the boundary of $\Delta$.
Here, $b'$ denotes the number of boundary vertices within $R^{\textsf o}$ plus the number of its bounding frames.
We call the union of such radial curves over all boundary-containing regions the \emph{skeleton forest}.

For a boundary-free annulus $R^{\textsf x}$ bounded by two frames $C$ and $C'$, observe that $R^{\textsf x}$ is a subregion of some thick maximal boundary-free ring $\ring(I_{i},I_{j})$ described above.
Precisely, $R^{\textsf x} \subseteq \ring(I_{i+30},I_{j-30}) \subset \ring(I_{i},I_{j})$.
Within $\ring(I_{i+30},I_{j-30})$, we compute a maximum set $\mathcal P$ of vertex-disjoint paths between $I_{i+30}$ and $I_{j-30}$.
From $\mathcal P$, we select the maximal subpaths that are entirely contained in $R^{\textsf x}$ and connect a vertex on $C$ to a vertex on $C'$.
We call these selected subpaths $\mathcal P'$ the \emph{reference paths}.
Since both $C$ and $C'$ have complexity $O(\tw(G))$, the size of $\mathcal P'$ is also bounded by $O(\tw(G))$.
It is clear that computing the entire skeleton forest and all reference paths requires $O(b \cdot \tw(G) \cdot n)$ time.
\subsection{Analysis of \Cref{sur_lem:summary_pdp}.}
Here, we fix a set $\widetilde{\mathcal T}$ of pairs of $\widetilde T$.
We allow $v$ to appear in $\widetilde{\mathcal T}$ at most $2x$ times, where each boundary vertex $v \in \widetilde T$ corresponds to $x$ boundary vertices in $G$.
Cho et al.~\cite{cho2023parameterized} designed a canonical procedure to obtain a weak $\widetilde{\mathcal T}$-linkage in the radial completion of $\widetilde G$ (if one exists) with respect to each boundary-free annulus and each boundary-containing region so that it satisfies \Cref{sur_lem:summary_pdp}.
Since $\widetilde G$ is embedded in the plane without holes, the canonical construction applies to $\widetilde G$ as well.
In the following, we analyze the complexity stated in \Cref{sur_lem:summary_pdp}.
Precisely, we analyze the number of discretely homotopic classes of weak linkages $\mathcal W$ satisfying the conditions in \Cref{sur_lem:summary_pdp} as follows, with respect to each boundary-free annulus and each boundary-containing region.
\begin{itemize}
    \item The walks in $\mathcal W$ are pairwise vertex-disjoint except on frames or skeleton forests,
    \item the walks in $\mathcal W$ traverse a boundary-free annulus along the reference paths, and
    \item the walks in $\mathcal W$ traverses an edge in the radial completion of $\widetilde G$ but not in $G$, then it is an edge on a frame or a skeleton forest. Additionally, such an edge is traversed at most $O(\tw(G))$ times in the walks.
\end{itemize}

\subparagraph*{Boundary-free annulus.}
Let $R^{\textsf x}$ be a boundary-free annulus bounded by two frames $C$ and $C'$.
We first show that the number of discretely homotopic classes of weak linkages $\mathcal W$ connecting $C$ and $C'$ within $R^{\textsf x}$ is at most $\textsf{poly}(\tw(G))$.
Note that $C$ and $C'$ are nooses of $\widetilde G$, and thus they are cycles of its radial completion.
For discretely homotopic classes, we can contract all edges of the radial completion of $\widetilde G$ lying on $C$ (and on $C'$) into a single loop.
Observe that even if the walks in $\mathcal W$ traverse $C$ (or $C'$) multiple times, they must all do so in the same direction, and the traversal between $C$ and $C'$ naturally proceeds along the reference paths in order.
Furthermore, the number of times two walks in $\mathcal W$ traverse $C$ (or $C'$) can differ by at most one.
Therefore, the discretely homotopic classes are determined by
(i) the number of times a walk traverses the contracted loop for $C$ and for $C'$, and
(ii) the size of the weak linkage $\mathcal W$.
Both are bounded by $O(\tw(G))$.
Therefore, the weak linkages can be encoded into at most $\textsf{poly}(\tw(G))$ discretely homotopic classes for each boundary-free annulus $R^{\textsf x}$.

\subparagraph*{Boundary-containing regions.}
For a boundary-containing region $R^{\textsf o}$, let $b'$ denote the number of boundary vertices of $\widetilde T$ contained in $R^{\textsf o}$ together with its bounding frames.
Recall that $\widetilde G$ is obtained by contracting the boundary vertices of $G$.
Moreover, if a boundary vertex $v \in \widetilde T$ corresponds to $x$ boundary vertices in $G$, then we allow $v$ to appear in a weak linkage at most $2x$ times.
For clarity, we further contract each bounding frame of $R^{\textsf o}$ into a single vertex.
The skeleton forest within $R^{\textsf o}$ consists of $O(b')$ radial curves, each of complexity at most $O(\tw(G))$, which implies that the skeleton forest has at most $O(b')$ vertices of degree one or at least three.
By cutting $R$ along the skeleton forest, we obtain a topological disk $\Delta$.
Additionally, by further contracting the maximal paths on the boundary of $\Delta$ consisting of degree-two vertices, the boundary of $\Delta$ consists of $O(b')$ vertices, and each boundary vertex of $\Delta$ corresponds to either
(i) a maximal path of degree-two vertices on the skeleton forest,
(ii) a vertex of degree at least three on the skeleton forest,
(iii) a boundary vertex within $R^{\textsf o}$, or
(iv) a bounding frame incident to $R^{\textsf o}$.
We assign each boundary vertex of $\Delta$ a weight equal to the number of identified vertices in $G$ it represents.

Then the discretely homotopic classes of weak linkages in $R^{\textsf o}$ can be encoded as a weighted diagonalization of $\Delta$, where for a boundary vertex $v$ of $\Delta$, the number of diagonals incident to $v$ is at most its weight.
This is because the weak linkages consist of pairwise vertex-disjoint walks inside $R^{\textsf o}$ (and hence inside $\Delta$).
The number of such diagonalizations is at most $2^{O(b'\log b')} (M_{\textsf o}+\tw(G))^{O(b')}=(M_{\textsf o}+\tw(G))^{O(b')}$.
Here, let $M_{\textsf o}$ denote twice the number of boundary vertices of the original graph $G$ contained in $R^{\textsf o}$. Note that $M_{\textsf o}$ is at least $b'$.

\medskip

Overall, the total number of discretely homotopic classes over all regions is $\textsf{poly}(\tw(G))^{O(b)}\cdot(M'+\tw(G))^{O(b)}=M^{O(b)}$, where $M$ denotes twice of the number of boundary vertices in the original graph $G$ plus its treewidth $\tw(G)$.
Furthermore, for each discretely homotopic class, a weak linkage in that class can be computed in linear time according to the above description. The detailed algorithm is analogous to~\cite{cho2023parameterized}.
Therefore, the complexity bounds stated in \Cref{sur_lem:summary_pdp} hold.

\bibliography{paper}

\end{document}